\documentclass[11pt]{amsart}

\usepackage{amssymb,amsmath,amssymb}
\usepackage{abstract}
\usepackage{amsthm}
\usepackage{mathrsfs}
\usepackage{graphicx}
\usepackage{subfigure}
\usepackage{caption}
\usepackage{float}
\usepackage{epstopdf}
\usepackage{makecell,rotating,multirow,diagbox}
\usepackage{booktabs}
\usepackage{makecell}
\usepackage{color}
\usepackage{url}
\usepackage[colorlinks,linkcolor=blue,anchorcolor=blue,citecolor=red]{hyperref}
\usepackage{appendix}
\usepackage{enumerate}
\usepackage{cite}
\usepackage{graphicx}
\usepackage{appendix}

\def\bm{\textbf{m}}
\def\d{\mathrm{d}}

\def\be{\textbf{e}}
\def\bx{\boldsymbol{x}}

\newcommand{\abs}[1]{\lvert#1\rvert}

\newtheorem{theorem}{Theorem}[section]
\newtheorem{assumption}{Assumption}

\newtheorem{lemma}{Lemma}[section]
\newtheorem{remark}{Remark}[section]
\newtheorem{proposition}{Proposition}[section]
\newtheorem{definition}{Definition}[section]

\begin{document}

	\title[Continuum Limit of Hexagon-located Spin Dynamics]{Continuum Limit of Spin Dynamics on Two-dimensional Hexagonal Lattice}
	\author{Mengjia Bai, Jingrun Chen, Zhiwei Sun, Yun Wang}
	\maketitle
	\begin{abstract}
  This study investigates the atomistic spin system in $\rm CrCl_{3}$, which exhibits topologically nontrivial meron structures within its layered hexagonal lattice framework. We analyze the complete model of discrete spin dynamics on a two-dimensional hexagonal lattice and demonstrate its convergence to the continuum Landau-Lifshitz-Gilbert equation in the weak sense.
The primary challenge lies in defining appropriate difference quotient and interpolation operators for the hexagonal lattice since the loss of symmetry.
To address these, we utilized a one-step difference quotient for the 2nd nearest neighbors and introduced novel multi-step difference quotients for the 1st and 3rd nearest neighbors, enabling the integration by parts formula. 
Additionally, we generalized Ladysenskaya's interpolation operator for hexagonal lattices and provided an alternative strategy for the convergence procedure by applying an isometric mapping property. 
This work provides necessary tools for analyzing weak convergence in other atomistic nonlinear problems on hexagonal lattices towards the continuum limit.

	\noindent{$\mathbf{Key\ words:}$ Spin dynamics, hexagonal lattice, Landau-Lifshitz-Gilbert equation, interpolation operator, difference method.  }
	
		
	\end{abstract}
	\section{Introduction}
	Nature favors the hexagonal lattice, seen in honeycombs created by insects and graphene, a cutting-edge technology. Nanotechnology has led to the discovery of (quasi) two-dimensional materials with intriguing electronic and magnetic properties, such as superconducting states in twisted graphene and meron spin textures in $\rm CrCl_{3}$, both possessing a hexagonal lattice. To understand the dynamics of physical quantities in these materials, two models have been developed: (1) an atomistic model that uses ordinary differential equations (ODEs) with degrees of freedom defined on the hexagonal lattice; (2) a phenomenological model that employs partial differential equations (PDEs), treating time and space variables continuously. While it is convenient to solve the phenomenological model, the atomistic model provides an accurate description of the underlying material. Therefore, it is natural to ask about the connection between atomistic and continuum models.
	
	There is an extensive literature addressing material properties from various perspectives, including mechanical \cite{blanc_molecular_2002, e_cauchyborn_2007}, electronic \cite{e_continuum_2007,e_electronic_2011,kohn_sham_equation_2013}, and magnetic properties \cite{sulem_continuous_1986,desimone_reduced_2002}, among others. For mechanical property, the consistency between atomistic and continuum models is derived in \cite{e_continuum_2007} with stability proven in \cite{e_cauchyborn_2007}. The electronic property aims to understand solids at the quantum or molecular level with rigorous derivation of continuum models established in \cite{e_continuum_2007}. Magnetic property involves deriving the continuum limit of a discrete spin system on a cubic lattice as well as identifying a physically relevant thin-film limit starting from a three-dimensional micromagnetic model  \cite{sulem_continuous_1986, desimone_reduced_2002}. 
 
 Additionally, mathematical works have been done to pass the limit from ODE systems on hexagonal lattices to continuum PDE models \cite{buranay_hexagonal_2020,dosiyev_approximation_2014}. However, these studies assume sufficient smoothness of the continuum solutions. This raises the question: if the solution of the continuum model exists only in the weak sense, is it reliable for describing the behavior of an atomistic system on a hexagonal lattice from a macroscopic scale?
 In this study, we address this question by focusing on the dynamics of discrete spins in $\rm CrCl_{3}$.
 The discrete spin model is known for its excitation of various topologically nontrivial structures like domain walls, skyrmions, and merons, 
Recently, merons, along with antimerons, have been demonstrated to exist in the magnetic structure of hexagonal lattice $\rm CrCl_{3}$ \cite{augustin_properties_2021}. In this work, we analyze the complete model of discrete spin dynamics on a two-dimensional hexagonal lattice,
 and establish a rigorous passage from the discrete spin system to the well-known Landau-Lifshitz-Gilbert (LLG) equation in the weak sense.

To present the main finding, we begin by introducing the spin model for $\rm CrCl_{3}$. The two-dimensional hexagonal lattice point is denoted by $\bx_i$, where $i\in \mathbb{N}$, and each atomic spin is represented by $\bm(\bx_i)$ with a unit length. The equilibrium state of spin $\bm(\bx_i)$ is determined by minimizing its energy. In the case of $\rm CrCl_{3}$, this can be expressed as per \cite{augustin_properties_2021}:
\begin{equation}
    \label{energy from reference}
		\begin{split}
			\mathcal{H}[\bm]
			=
			\mathcal{H}_{\mathrm{e}}[\bm]
   +
   \mathcal{H}_{\mathrm{b}}[\bm]
   +
   \mathcal{H}_{\mathrm{a}}[\bm]
   +
   \mathcal{H}_{\mathrm{z}}[\bm].
		\end{split}
\end{equation}
The different contributions to the energy in \eqref{energy from reference} are:
\begin{enumerate}[1.]
    \item \textbf{Bilinear Exchange Energy}. The fundamental property of ferromagnetic materials is that the spins experience the presence of an exchange field, decribed by
    \begin{equation*}
    \begin{aligned}
        \mathcal{H}_{\mathrm{e}}[\bm]
        = &
        -\frac{1}{2}\sum_{i,j\in \mathbb{N}}J_{ij} \bm(\boldsymbol{x}_{i})\cdot\bm(\boldsymbol{x}_{j}).
   \end{aligned}
   \end{equation*}
   Here, $\bm(\boldsymbol{x}_{i})$ and $\bm(\boldsymbol{x}_{j})$ are localized magnetic moments on $\rm Cr$ atomic sites $i$ and $j$ coupled by pair-wise exchange interactions. It is usually considered up to the third nearest neighbour for $J_{ij}\;(J_{1}-J_{2}-J_{3})$, where $J_k$ is constant denote the magnetic quantities of the $k$th nearest neighbour; see Figure \ref{Fig4}. Here, we use the notation $J_{ij}\;(J_{1}-J_{2}-J_{3})$ to represent the following: suppose $\boldsymbol{x}_{i}$ is the $k$th nearest neighbour of $\boldsymbol{x}_{j}$, then
   \begin{equation*}\label{define J}
       J_{ij} = 
       \left\{\begin{aligned}
           J_{k},&\qquad \text{for $k=1,2,3$};\\
           0,&\qquad \text{for $k>3$}.
       \end{aligned}\right.
   \end{equation*}

      \item \textbf{Biquadratic (BQ) Exchange Energy}.
    BQ exchange interactions  involving the hopping of
two electrons are critical in the elucidation of the magnetic features of layered magnetic materials \cite{kartsev_biquadratic_2020}. This is described by
    \begin{equation*}
        \begin{aligned}
            \mathcal{H}_{\mathrm{b}}[\bm]
            =
            -\frac{1}{2}
			\sum_{i,j\in \mathbb{N}} K_{ij}\big(\bm(\boldsymbol{x}_{i})\cdot\bm(\boldsymbol{x}_{j})\big)^{2}.
        \end{aligned}
    \end{equation*}
    Its strength is given by the constant $K_{ij}\ (K)$ where  only the nearest neighbours taken into account, i.e., 
    suppose $\boldsymbol{x}_{i}$ is the $k$th nearest neighbour of $\boldsymbol{x}_{j}$, then
   \begin{equation*}
       K_{ij} = 
       \left\{\begin{aligned}
           K,&\qquad \text{for $k=1$};\\
           0,&\qquad \text{for $k>1$}.
       \end{aligned}\right.
   \end{equation*}
    This is the simplest and most natural form of non-Heisenberg coupling. \\

    \item \textbf{Anisotropy Energy}. The anisotropy energy can  be described by the terms of the form
    \begin{equation*}
        \begin{aligned}
            \mathcal{H}_{\mathrm{a}}[\bm]
            = &
            -
			\sum_{i\in \mathbb{N}}\lambda_{i}\big(\bm(\boldsymbol{x}_{i})\cdot \mathbf{e}^{x}\big)^{2}\\
   &- \frac{1}{2}
			\sum_{i,j\in \mathbb{N}}L_{ij}\big(\bm(\boldsymbol{x}_{i})\cdot\mathbf{e}^{z}\big)\big(\bm(\boldsymbol{x}_{j})\cdot\mathbf{e}^{z}\big).
        \end{aligned}
    \end{equation*}
    Here, the first term is induced by electronic structure of the underlying crystalline lattice with a preferred orientation $\mathbf{e}^{x}$. We assume $\lambda_{i}$ is the single-ion magnetic anisotropy coefficient. The second term is the symmetric anisotropy exchange energy that can be deduced when considering relativistic effect. It is considered up to the third nearest neighbour for $L_{ij}\;(L_{1}^{(ij)}-L_{2}^{(ij)}-L_{3}^{(ij)})$, where $L_k^{(ij)}\;(k=1,2,3)$ denote the magnetic quantities of the $k$th nearest neighbour. More precisely, suppose $\boldsymbol{x}_{i}$ is the $k$th nearest neighbour of $\boldsymbol{x}_{j}$, then we have:
    \begin{equation*}
       L_{ij} = 
       \left\{\begin{aligned}
           L_{k}^{(ij)},&\qquad \text{for $k=1,2,3$};\\
           0,&\qquad \text{for $k>3$}.
       \end{aligned}\right.
   \end{equation*}
    \\
    
    \item \textbf{Zeeman Energy}. In the presence of an external magnetic field $\mathbf{B}_i$, the magnetization tends to align with it. 
This translates into an energy term of the form
 \begin{equation*}
        \begin{aligned}
            \mathcal{H}_{\mathrm{z}}[\bm]
            =
            -\sum_{i\in\mathbb{N}}\mu\bm(\boldsymbol{x}_{i})\cdot\mathbf{B}_i,
        \end{aligned}
    \end{equation*}
    where $\mu$ is the magnetic permeability of vacuum.
\end{enumerate}
Furthermore, note that we consider the hexagonal lattice over $\mathbb{R}^2$ without the compactness, the following assumption on the coefficients are needed such that the anisotropy energy and Zeeman energy are finite.
\begin{assumption}\label{assume}
    We assume that there exists continuous functions $\lambda(\boldsymbol{x})$, $L_k(\boldsymbol{x})$, $k=1,2,3$ and $\mathbf{B}(\boldsymbol{x})$ which are $L^1(\mathbb{R}^2)$ bounded. Furthermore, it holds
    \begin{equation*}
        \lambda_i = \lambda (\boldsymbol{x}_i),
        \qquad 
        L_k^{(ij)} = \frac{1}{2}\Big(L_k(\boldsymbol{x}_i) + L_k(\boldsymbol{x}_j)\Big)
        \qquad 
        \mathbf{B}_i = \mathbf{B} (\boldsymbol{x}_i).
    \end{equation*}
\end{assumption}
	
	The time evolution of each atomic spin from any possible initial structure to its equilibrium state follows a discrete spin model that described by the Landau-Lifshitz-Gilbert (LLG) equation:
	\begin{equation}\label{llg equation on lattice}
		\frac{\d\bm}{\d t}
		=
		-\frac{\gamma}{1+\alpha^{2}}\Big\{\bm
		\times\mathbf{B}_{\rm{eff}}
		+
		\alpha\bm\times \big(\bm\times\mathbf{B_{\rm{eff}}}\big) \Big\}.
	\end{equation}
	Here, $\gamma>0$ represents the gyromagnetic constant, and $\alpha>0$ denotes the damping constant. The effective field, denoted as $\mathbf{B_{\rm{eff}}}$, is computed by
	\begin{equation}\label{discrete effect field on lattice}
		\mathbf{B_{\rm{eff}}}=-\frac{\delta \mathcal{H}}{\delta \bm}.
	\end{equation}
	
	\begin{figure}[htbp]
		\centering
		\includegraphics[scale=1.2]{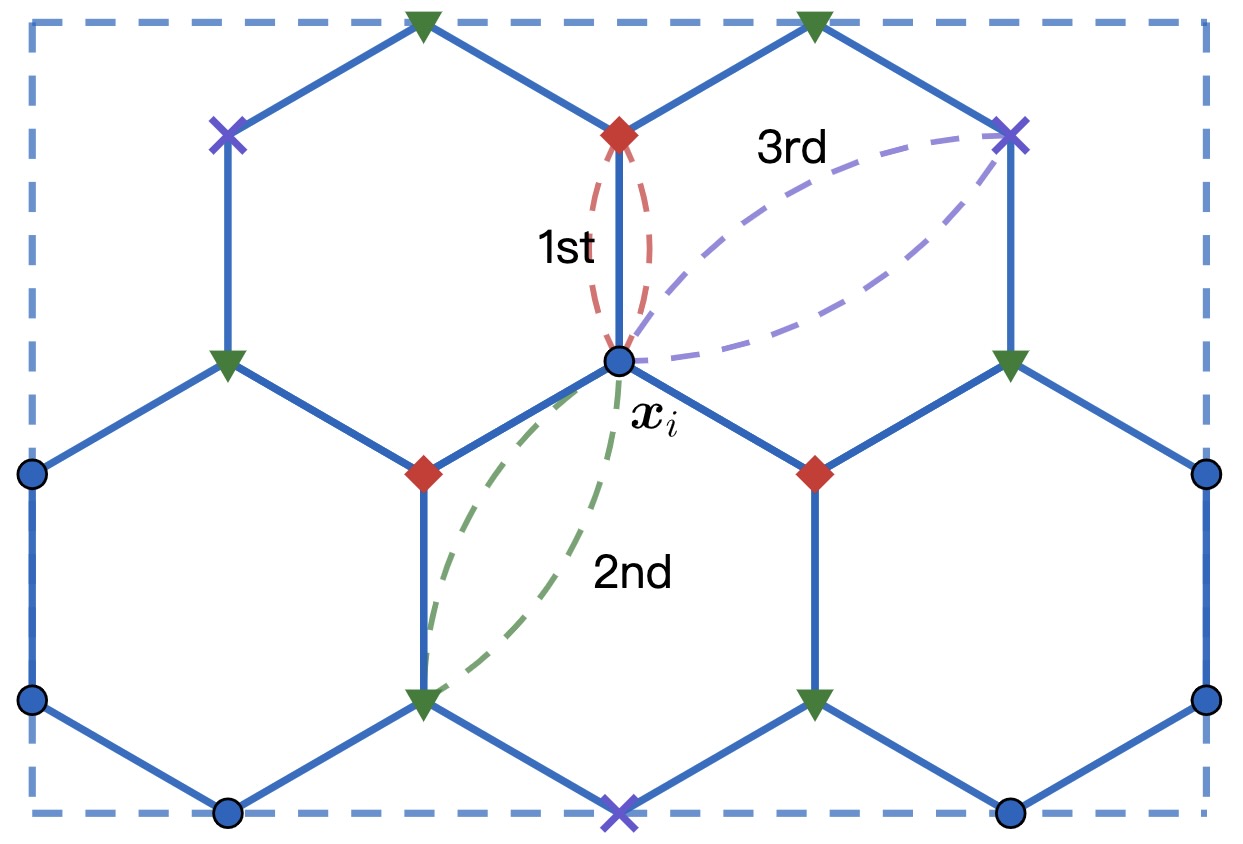}
		\caption{Illustration of the hexagonal lattice. Red squares, green triangles, and purple crosses, represent the 1st, 2nd, and 3rd-neighbor nodes of the reference node $\boldsymbol{x}_i$ with $i\in \mathbb{N}$.
		}
		\label{Fig4}
	\end{figure}

In an early work \cite{sulem_continuous_1986}, Sulem, Sulem, and Bardos analyzed the system \eqref{energy from reference}-\eqref{discrete effect field on lattice} on a square grid. They considered only the exchange energy in \eqref{energy from reference} and interactions up to the $1$st nearest neighbor, without including the damping term ($\alpha=0$). By utilizing Ladysenskaya's interpolation operator, they demonstrated that the solution to the discrete spin model can be regarded as an approximation of the weak solution of the continuum LLG equation. In this case, it was found that the exchange term corresponds to Laplacian. A similar outcome is obtained in \cite{alouges_global_1992} when accounting for damping. Furthermore, the results of 1D using the difference method are presented in \cite{zhou1991existence}. We refer to \cite{sulem_continuous_1986} which shows that the solution to the discrete spin model converges to the smooth solution of the continuum LLG equation if the initial data for \eqref{energy from reference}-\eqref{discrete effect field on lattice} is sufficiently smooth. However, while only a simple model on a square lattice was considered in \cite{sulem_continuous_1986}, we study a more ``physical" model with the full system \eqref{energy from reference}-\eqref{discrete effect field on lattice} on a two-dimensional hexagonal lattice. Our main objective is to consider the Cauchy problem for LLG system \eqref{energy from reference}-\eqref{discrete effect field on lattice}, and investigate its weak convergence as the rescaled lattice constant of the hexagonal lattice approaches $0$ towards  degenerated nonlinear equations given by  \eqref{Continuous equation}-\eqref{continuous eff field}.

The main concept is based on \cite{sulem_continuous_1986,ladyzhenskaya_boundary_1985}, but our focus is on a hexagonal lattice, which differs from the classical square lattice discussed in those works. Hence we need to introduce new difference quotients and interpolation operators specific to the hexagonal lattice and study their corresponding properties. 
The primary challenge lies in defining appropriate difference quotient and interpolation operators for the hexagonal lattice since the \textbf{loss of symmetry}, especially when considering the $1$st and $3$rd nearest neighborhoods.
Indeed, each node in square lattice is a 4-fold rotation center, which possesses the rotation symmetry operation as a cyclic group $C_4$. While each node in hexagonal lattice is only a 3-fold rotation center. This loss of symmetry is caused by the fact that the hexagonal lattice can be viewed as a union of two triangular lattices staggered from each other, denoted by $\mathcal{G} = \mathcal{G}_1 \cup \mathcal{G}_2$ (see definition in \eqref{define G1 G2}) where $\mathcal{G}$ represents the nodes of hexagonal lattice.
As a consequence, we find a one-step difference quotient can be used to 
conduct the convergence to the continuum operator for $2$nd nearest neighborhoods, since the nodes all appear in only one of the two sets, $\mathcal{G}_1$ and $\mathcal{G}_2$. However, when considering the $1$st and $3$rd nearest neighborhoods, the one-step difference quotient fails in integration by parts.
In order to overcome it, we construct several new multi-step difference quotients. Overall, we introduce the procedure as follows:
\begin{itemize}
    \item[-] Step 1. Defining the one-step difference quotients for $2$st nearest neighborhoods.
    \item[-] Step 2. Constructing new multi-step difference quotients for the $1$st and $3$rd nearest neighborhoods, which yield the formula of integration by parts.
    \item[-] Step 3. Introducing some auxiliary difference quotients to establish the connection between the discrete Laplace operator and our new multi-step difference quotient in Step 2.
\end{itemize}
The above argument of construction can be applicable to other parabolic problems on a hexagonal lattice. 
On the other hand, interpolation operators play an important role in connecting solutions of discrete and continuum models. As shown in \cite{ladyzhenskaya_boundary_1985}, in order to conduct the weak convergence argument, one needs three different interpolation operators which map the discrete function into step function, piece-wise linear function, and hybrid function that is stepped in one dimension and piece-wise linear in other dimensions. However, constructing the last one on hexagonal lattice is difficult. To address this problem, we provide an alternative strategy to prove the convergence by using only two interpolation operators. 
Parallel to the work in \cite{ladyzhenskaya_boundary_1985}, by utilizing the difference quotients and interpolation operators constructed in this work, we finally derive the convergence from the discrete spin model on hexagonal lattice to the continuum LLG equation.
As a by-product, our result provides a precise contribution on physical parameters in the atomistic model \eqref{energy from reference}
to the continuum LLG equation, as listed in Table \ref{Table 3}. 
Of particular interest, the bilinear exchange constant $J_k$ (for $k=1,2,3$) and the biquadratic exchange constant $K$ collectively form a continuum exchange constant given by $\frac{3}{4}J_{1} +\frac{9}{2}J_{2} +3J_{3} +\frac{3}{2}K$.

	\subsection{Parameterized Cauchy problem}
	To begin, we will establish the Cauchy problem for system \eqref{energy from reference}-\eqref{discrete effect field on lattice}, with the hexagonal lattice constant being parameterized. We denote $h^*$ as the dimensionless ratio between the lattice constant and characteristic length of the system, and $h$ as a small quantity satisfying $0\le h\le h^*$. Let $\bm_{h}$ be a discrete function defined on a hexagonal lattice with grid size $h$. Using these parameters, we can express the energy in a parameterized form:
   \begin{equation}\label{discrete energy}
		\begin{aligned}
			\mathcal{H}_h[\bm_h]
			=
			& \frac{(h^*)^2}{h^2}
   \Big( \mathcal{H}_{\mathrm{e}}[\bm]
   +
   \mathcal{H}_{\mathrm{b}}[\bm] \Big)
   +
   \mathcal{H}_{\mathrm{a}}[\bm]
   +
   \mathcal{H}_{\mathrm{z}}[\bm].
		\end{aligned}
	\end{equation}
	It follows that $\mathcal{H}_h[\bm_h] = \mathcal{H}[\bm_h]$ when $h=h^*$.
	The corresponding effective field $\mathbf{B}_{\rm{eff},h}(\boldsymbol{x}_{h})$ can be calculated by
	\begin{equation}\label{discrete effect field}
		\mathbf{B_{\rm{eff},h}}=-\frac{\delta \mathcal{H}_h}{\delta \bm_{h}}.
	\end{equation}
	Moreover, rescaling time as
	\begin{equation*}
		t^*=\frac{\gamma (h^*)^{2} }{1+\alpha^{2}} t,
	\end{equation*}
	and prescribing the initial data, then we set the Cauchy problem for LLG system \eqref{llg equation on lattice} with parameter $h$ by
	\begin{equation}\label{discrete LLG equation in tro}
		\left\{\begin{aligned}
			\frac{\d\bm_h}{\d t^*}(\boldsymbol{x}_{h})
			= &
			-\frac{1}{(h^*)^{2}}\Big\{\bm_{h}(\boldsymbol{x}_{h})
			\times\mathbf{B}_{\rm{eff},h}(\boldsymbol{x}_{h})\\
			& \qquad \qquad \quad +
			\alpha\bm_{h}(\boldsymbol{x}_{h})\times \big(\bm_{h}(\boldsymbol{x}_{h})\times\mathbf{B_{\rm{eff},h}}(\boldsymbol{x}_{h})\big) \Big\},\\
			\bm_{h}(\boldsymbol{x}_{h},0)
			=& \bm_{\rm{init},h}(\boldsymbol{x}_{h}),\quad
			\left|\bm_{\rm{init},h}(\boldsymbol{x}_{h}) \right| =1.
		\end{aligned}\right.
	\end{equation}
As we are interested in phenomena occurring at large scales, compared to the lattice mesh size, we take the limit of \eqref{discrete LLG equation in tro} as $h$ approaches zero.

	\subsection{Main result}
	We denote $H^s$ (and $L^p$) for the Sobolev space $H^s(\mathbb{R}^2;\mathbb{R}^3)$ with $s\geq 1$ (and the Lebesgue space $L^p(\mathbb{R}^2;\mathbb{R}^3)$ with $p\in [1,\infty]$, respectively). Additionally, we denote $\dot{H}^s$ as the homogeneous Sobolev space whose seminorm is determined through Fourier transform:
	\begin{equation*}
		\|u\|^2_{\dot{H}^s}
		:=
		\frac{1}{2\pi}\int_{\mathbb{R}^2}|\hat{u}(\boldsymbol{\xi})|^2|\boldsymbol{\xi}|^{2s}\mathrm{d}\boldsymbol{\xi},\quad
		\mbox{where}
		\quad
		\hat{u}(\boldsymbol{\xi})
		=
		\int_{\mathbb{R}^2}e^{-i\boldsymbol{x}\cdot\boldsymbol{\xi}}u(\boldsymbol{x})\mathrm{d}\boldsymbol{x}.
		\end{equation*}\
	In particular, $\|u\|_{\dot{H}^1}=\|\nabla u\|_{L^2}$.
	
	Formally, taking the limit as $h \rightarrow 0$ for the discrete system given by equations \eqref{discrete energy}-\eqref{discrete LLG equation in tro} yields the following continuum LLG equation:
	\begin{equation}\label{Continuous equation}
		\left\{
		\begin{aligned}
			&\frac{\partial\bm}{\partial t}
			=
			-\bm\times\mathbf{B_{\rm{eff}}}
			-\alpha\bm\times(\bm\times\mathbf{B_{\rm{eff}}}), \\
			&\bm({\boldsymbol{x}},0)
			=
			\bm_{\rm{init}}({\boldsymbol{x}}), \quad
			\abs{\bm_{\rm{init}}({\boldsymbol{x}})} = 1.
		\end{aligned}
		\right.
	\end{equation}
Here the continuum effective field $\mathbf{B}_{\rm{eff}}$ is given by
\begin{align}\label{continuous eff field}
	\begin{split}
		\mathbf{B}_{\rm{eff}}
		=
		J^*\Delta\bm
		+L^*(\bm\cdot\mathbf{e}^{z})\mathbf{e}^{z}
		+\lambda^*(\bm\cdot\mathbf{e}^{x})\mathbf{e}^{x}
		+\mu^*\mathbf{B}.
	\end{split}
\end{align}	
The coefficients $J^*$, $L^*$, $\lambda^*$, and $\mu^*$ are precisely linked to the parameters in the discrete model \eqref{discrete energy} as listed in Table \ref{Table 3}.
\begin{table}[H]
	\centering
	\caption{Relationship between parameters in the atomistic and continuum Landau-Lifshitz-Gilbert equations. $h^*$ is the dimensionless lattice constant.}
	\label{Table 3}
	\begin{tabular}{|c|c|}
		\hline
		Parameter &  Relationship \\ \hline
		Exchange & $J^* = \frac{3}{4}J_{1} +\frac{9}{2}J_{2} +3J_{3} +\frac{3}{2}K$\\\hline
		Anisotropy I & $L^* = (h^*)^{-2} (L_{1}+L_{2}+L_{3})$\\\hline
		Anisotropy II& $\lambda^* = (h^*)^{-2} \lambda$\\\hline
		External & $\mu^* = (h^*)^{-2} \mu$ \\\hline
	\end{tabular}
\end{table}
	The weak solution to \eqref{Continuous equation}-\eqref{continuous eff field} is defined in the following sense:
	\begin{definition}
		Suppose $\abs{\bm(\boldsymbol{x},t)} = 1$ a.e. in $\mathbb{R}^2\times [0,T]$. The function $\boldsymbol{m}$ is said to be a weak solution of \eqref{Continuous equation}-\eqref{continuous eff field}, if
		\begin{equation*}
			\begin{aligned}
				\nabla\bm \in L^\infty(0,T; L^2),\quad
				\frac{\partial \bm}{\partial t} \in L^2(0,T; L^2),\quad
				\nabla\cdot (\bm\times \nabla \bm) \in L^2(0,T; L^2),
			\end{aligned}
		\end{equation*}
		and there holds for any $\boldsymbol{\varphi}\in C_c^\infty (\mathbb{R}^{2}\times(0,T))$ that
			\begin{equation}\label{define weak solution}
			\begin{aligned}
			\int_{0}^{T}\int_{\mathbb{R}^{2}} \partial_{t}\bm\cdot \boldsymbol{\varphi}
			\,\mathrm{d}\boldsymbol{x}\, \mathrm{d} t
			= &
			- \int_{0}^{T}\int_{\mathbb{R}^{2}} \bm \times \mathbf{B}_{\rm{eff}}
			\cdot \boldsymbol{\varphi} \,\mathrm{d}\boldsymbol{x}\, \mathrm{d} t\\
			& -
			\alpha \int_{0}^{T}\int_{\mathbb{R}^{2}} \bm\times(\bm\times\mathbf{B_{\rm{eff}}})\cdot \boldsymbol{\varphi} \,\mathrm{d}\boldsymbol{x} \,\mathrm{d} t
			\end{aligned}
		\end{equation}
	with initial value $\bm(x,0) = \bm_{\rm{init}}(x)$.
	\end{definition}
Note that the definition above is reasonable because $\bm \times \Delta\bm =\nabla\cdot(\bm \times\nabla\bm)\in L^2(0,T; L^2)$.
	
	Now let us state the main results.
		\begin{theorem}(Convergence). \label{main theorem}
		Suppose the Assumption \eqref{assume} holds. Assume the initial data satisfies $\nabla\bm_{\rm{init}}(\boldsymbol{x})\in L^{2}(\mathbb{R}^{2})$ in equation \eqref{Continuous equation}.
  Given $\bm_{\rm{init},h}(\boldsymbol{x})\in \widetilde{H}^{1}_{loc,h} (\mathcal{G})$ and $p_{h}\bm_{\rm{init},h} \rightarrow\bm_{\rm{init}}$ strongly in $\dot{H}^{1} (\mathbb{R}^2)$, the discrete LLG equation \eqref{discrete energy}-\eqref{discrete LLG equation in tro} admits a unique solution $\bm_h$. The solution $\bm_h$ converges to a  weak solution $\bm^*$ to the continuum LLG equation \eqref{Continuous equation}-\eqref{continuous eff field} in the following way,
		\begin{equation*}
			\begin{gathered}
				p_h \bm_h,\ q_h\bm_h \rightarrow \bm^* \mbox{ strongly in $C([0,T]; L^2_{loc} (\mathbb{R}^2))$,} \\	
			p_h \bm_h \rightharpoonup \bm^* \mbox{ weakly in $L^\infty(0,T;\dot{H}^1 (\mathbb{R}^2))$.}
			\end{gathered}
		\end{equation*}
	
\end{theorem}
Here we note that $q_h,\ p_h$ are linear interpolation operators on a hexagonal lattice,  whose definitions and properties will be introduced in Section \ref{sec:interplolation to Hexagonal Lattice} and Propositions \ref{define qh} and \ref{define ph}. For a vector-valued function $\mathbf{m}_h$, $q_h \mathbf{m}_h,\ p_h \mathbf{m}_h$ are defined to be $ (q_h m_h^1,\, q_h m_h^2,\, q_h m_h^3)$ and $ (p_h m_h^1,\, p_h m_h^2,\, p_h m_h^3)$ respectively.  


	\begin{remark}
		The uniqueness of the weak solution obtained in Theorem \ref{main theorem} is unclear, as discussed in \cite{alouges_global_1992}.
		However, if a unique strong solution $\bm\in C^3(\mathbb{R}^2\times[0,T])$ exists, then it can be shown that $\bm^*=\bm$, which is referred to as ``weak-strong uniqueness" \cite{DIFRATTA2020103122}.
	\end{remark}

This article is organized as follows. Section \ref{sec: Notation on hexagonal} introduces the notation of difference quotients on a hexagonal lattice and derives the discrete formula for integration by parts. In Section \ref{sec: proof of convergence}, we deduce the weak convergence procedure and prove Theorem \ref{main theorem} using a priori energy estimate, leaving the calculation for later. Section \ref{sec:interplolation to Hexagonal Lattice} presents the construction and proof of properties of proper interpolation operators on a hexagonal lattice. Finally, in the Appendix, we transform the LLG system \eqref{discrete LLG equation in tro} into a discrete PDE type equation and derive uniform estimates for discrete solutions.

	\section{Notations on hexagonal lattice}\label{sec: Notation on hexagonal}
	Let $\mathcal{G}$ denote the collection of all the nodes of hexagonal lattice, with the parameterized side length $h>0$.
	Note that it takes account up to three layers of nearest neighbors of point-to-point interactions in the energy \eqref{energy from reference}. For any point $\bx_h\in \mathcal{G}$, we use $N_{k}$, $k=1,2,3$, to represent the number of nearest neighbors of $\bx_h$ in the $k$th layer, with
	\begin{align*}
		N_{1}=3,\ \ N_{2}=6,\ \ N_{3}=3,
	\end{align*}
	and write the $j$th neighbor in the $k$th layer of $\boldsymbol{x}_{h}$ by $\boldsymbol{x}_{h}^{k,j}$ where $j=1,\dots, N_{k}$. Figure \ref{Fig1} shows the arrangement mode of $\boldsymbol{x}_{h}^{k,j}$ in relation to $\bx_h$.
	It is worth noting that the points in $\mathcal{G}$ can be classified into two categories based on the distribution of their 1st-layer nearest neighbors: regular triangle type (Figure \ref{Fig1}(a)) and inverted triangle type (Figure \ref{Fig1}(b)). Let us write $\mathcal{G} = \mathcal{G}_1 \cup \mathcal{G}_2$, where
 \begin{equation}\label{define G1 G2}
	\begin{aligned}
		\mathcal{G}_{1}&=\big\{\boldsymbol{x}_{h}\in\mathcal{G}\ \big|\ \mbox{the nearest neighbors of $\boldsymbol{x}_{h}$ form a regular triangle }\big\},\\
		\mathcal{G}_{2}&=\big\{\boldsymbol{x}_{h}\in\mathcal{G}\ \big|\ \mbox{the nearest neighbors of $\boldsymbol{x}_{h}$ form an inverted triangle }\big\}.
	\end{aligned}
 \end{equation}
	
	\begin{figure}[htbp]
		\centering
		\includegraphics[scale=0.4]{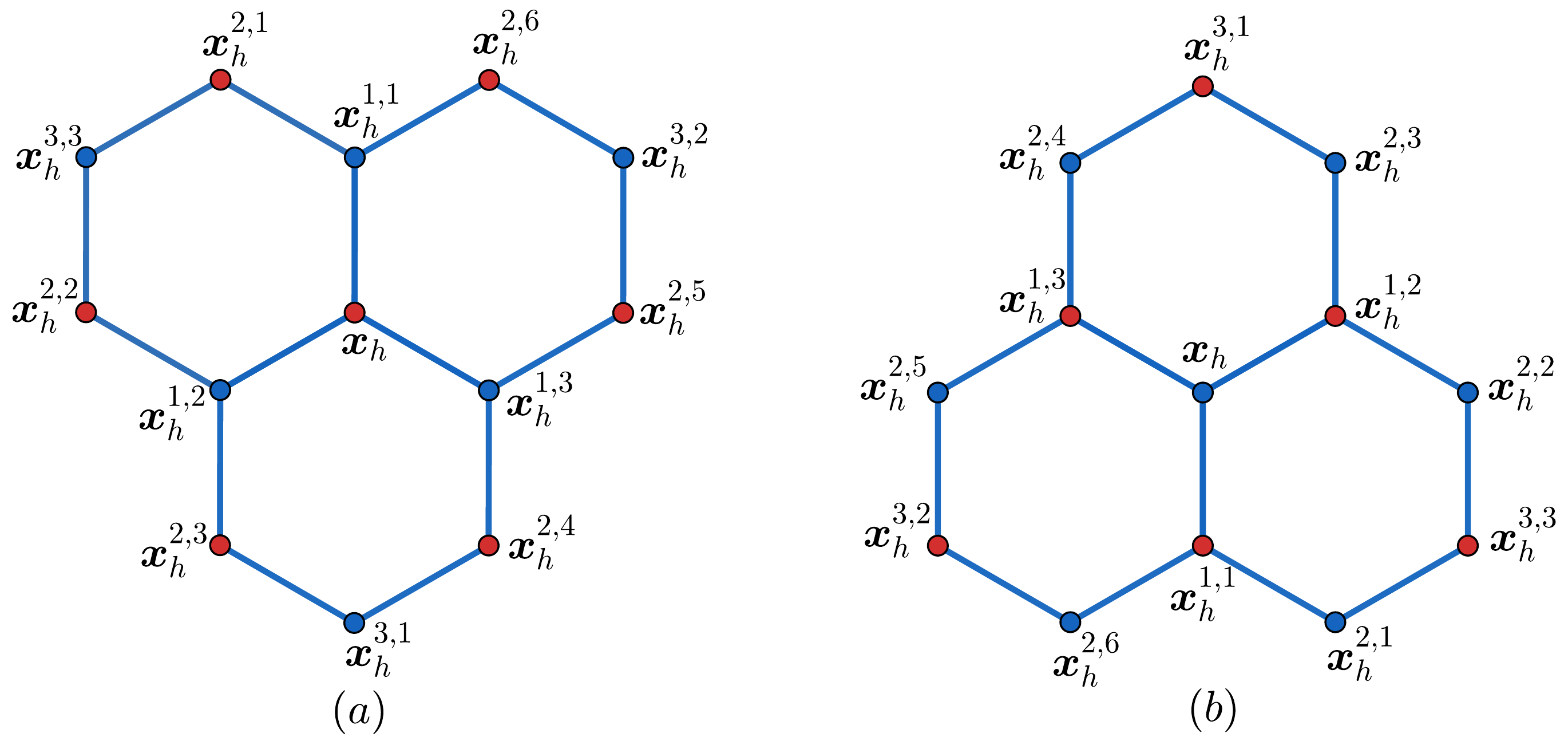}
		\caption{The $1$st, $2$nd, and $3$rd nearest neighbors of $\boldsymbol{x}_h\in \mathcal{G}$ with two categories: (a), $\boldsymbol{x}_h\in\mathcal{G}_1$; (b), $\boldsymbol{x}_h\in\mathcal{G}_2$}.
		\label{Fig1}
	\end{figure}

	\subsection{One-step difference quotients and discrete space}\label{sec: diff operator}
	Let us introduce the difference quotients and Laplace operator on a hexagonal lattice. We can define the one-step difference quotient for the $k$th nearest neighbors and $j$th direction as follows:
	\begin{equation}\label{approximations of derivatives1}
		\begin{aligned}
			\widetilde{D}^{k,j}u_{h}(\boldsymbol{x}_{h})
			=&
			\frac{u_{h}(\boldsymbol{x}_{h}^{k,j})-u_{h}(\boldsymbol{x}_{h})}{d_k h}\quad j=1,\dots, N_{k},\quad k =1,2,3.
		\end{aligned}
	\end{equation}
	Here the normalized constants $d_1 = 1$, $d_2=\sqrt{3}$, and $d_3 = 2$. The notation of backward difference quotient is necessary for the formula of integration by parts. Define the backward difference quotient
	\begin{equation}\label{calculate of backword}
		\widetilde{D}_{B}^{k,j}
		=
		\left\{ \begin{aligned}
			- \widetilde{D}^{k,j} , &\quad \mbox{if $k=1,3$}\\
			- \widetilde{D}^{k,j^*}, &\quad \mbox{if $k=2$}
		\end{aligned} \right.
	\end{equation}
	with $j^* = (j+3)\mod 6$. It follows by the definition that
	\begin{align}\label{exchange express of difference}
		\widetilde{D}^{k,j} u_{h}(\boldsymbol{x}_{h})
		=
		\widetilde{D}_{B}^{k,j} u_{h}(\boldsymbol{x}_{h}^{k,j}).
	\end{align}
	The discrete Laplace operator in the $k$th nearest neighbors is given by
	\begin{equation}\label{approximations of laplace}
		\begin{aligned}
			\widetilde{\Delta}_{k}u_{h}(\boldsymbol{x}_{h})
			=&
			\sum_{j=1}^{N_{k}}\frac{u_{h}(\boldsymbol{x}_{h}^{k,j}) -u_{h}(\boldsymbol{x}_{h})}{l_k h^{2}}
		\end{aligned}
	\end{equation}
	with $l_1 = 3/4$, $l_2 = 9/2$,  and $l_3 = 3$.
	The gradient operator in the $k$th nearest neighbors on hexagonal lattice is denoted by
	\begin{equation}\label{define nabla}
		\widetilde{\nabla}_{k}=g_k
		(\widetilde{D}^{k,1}, \widetilde{D}^{k,2},\cdots, \widetilde{D}^{k,N_k})
	\end{equation}
	with $g_1 = g_3 =\sqrt{6}/3$, $g_2 = \sqrt{3}/3$.
	
	We introduce some function spaces on a hexagonal lattice. The Hilbert space with an inner product given by
	\begin{equation}\label{inner product}
		(u_{h},v_{h})_{h}
		=
		\frac{3\sqrt{3}}{4}h^{2}\sum_{\boldsymbol{x}_{h}\in\mathcal{G}}u_{h}(\boldsymbol{x}_{h})\cdot
		v_{h}(\boldsymbol{x}_{h}).
	\end{equation}
	And $\widetilde{H}_{h}^1$ denotes the Hilbert space with inner product given by
	\begin{equation*}
		(u_{h},v_{h})_{\widetilde{H}_{h}^{1}}
		=
		(u_{h},v_{h})_{h}
		+
		(\widetilde{\nabla}_{1}u_{h},\widetilde{\nabla}_{1}v_{h})_{h}.
	\end{equation*}

	\subsection{Integration by parts for one-step difference quotients}\label{sec: Properties of hexagonal lattice}
The integration by parts formula is crucial in energy estimation. The following theorems indicate the one-step difference quotients satisfy the formula of integration by parts and Green's formula on hexagonal lattice.
\begin{lemma}\label{integration by parts for one step}
	For every function $u_{h}(\boldsymbol{x}_{h})$, $v_{h}(\boldsymbol{x}_{h})$ defined on $\mathcal{G}$, it holds that
		\begin{align}\label{inte by parts}
		u_{h}(\boldsymbol{x}_{h}) \cdot \widetilde{D}^{k,j} v_{h}(\boldsymbol{x}_{h})
		+
		\widetilde{D}_{B}^{k,j} u_{h}(\boldsymbol{x}_{h}^{k,j}) \cdot v_{h}(\boldsymbol{x}_{h}^{k,j})
		=
		\widetilde{D}^{k,j} \big\{ u_{h}(\boldsymbol{x}_{h}) \cdot v_{h}(\boldsymbol{x}_{h}) \big\}
	\end{align}
	and
	\begin{equation}\label{Green fomular}
		\widetilde{\Delta}_{k}u_{h}(\boldsymbol{x}_{h})
		=
		\frac{d_k^2}{2l_k}
		\sum_{j=1}^{N_k} \widetilde{D}^{k,j} \big\{ \widetilde{D}_{B}^{k,j} u_{h}(\boldsymbol{x}_{h}) \big\},
	\end{equation}
where $k=1,2,3$ and $j=1,\dots, N_{k}$.
\end{lemma}
\begin{proof}
	The formula \eqref{inte by parts} follows directly from the definition \eqref{approximations of derivatives1} and the equation \eqref{exchange express of difference}.
	Now let us prove the formula \eqref{Green fomular}. It follows from \eqref{approximations of derivatives1} and \eqref{exchange express of difference} that
	\begin{equation*}
		\begin{aligned}
			\widetilde{D}^{k,j} \big\{ D_{B}^{k,j} u_{h}(\boldsymbol{x}_{h}) \big\}
			=&
			\frac{1}{d_k h} \big\{ \widetilde{D}_{B}^{k,j} u_{h}(\boldsymbol{x}_{h}^{k,j})
			-
			\widetilde{D}_{B}^{k,j} u_{h}(\boldsymbol{x}_{h})  \big\}\\
			=&
			\frac{1}{d_k h} \big\{ \widetilde{D}^{k,j} u_{h}(\boldsymbol{x}_{h})
			-
			\widetilde{D}_{B}^{k,j} u_{h}(\boldsymbol{x}_{h})  \big\}.
		\end{aligned}
	\end{equation*}
	Summing up for $j = 1,\dots, N_k$ yields
	\begin{equation}\label{D DB m}
		\begin{aligned}
			\sum_{j=1}^{N_{k}} \widetilde{D}^{k,j} \big\{ \widetilde{D}_{B}^{k,j} u_{h}(\boldsymbol{x}_{h}) \big\}
			=&
			\frac{2}{d_k h} \sum_{j=1}^{N_{k}} \widetilde{D}^{k,j} u_{h}(\boldsymbol{x}_{h}) .
		\end{aligned}
	\end{equation}
	On the other hand, it follows from \eqref{approximations of laplace}
	\begin{equation}\label{another form of delta}
		\begin{aligned}
			\widetilde{\Delta}_{k}u_{h}(\boldsymbol{x}_{h})
			=&
			\frac{d_k }{l_k h} \sum_{j=1}^{N_{k}} \widetilde{D}^{k,j} u_{h}(\boldsymbol{x}_{h}).
		\end{aligned}
	\end{equation}
	Hence the formula \eqref{Green fomular} is proved.
\end{proof}
Since $\sum_{\bx_h \in \mathcal{G}} \widetilde{D}^{k,j} f_h(\bx_h) = 0$, the equation \eqref{inte by parts} indicates the formula of integration by parts
\begin{equation*}
	\big(u_{h},\, \widetilde{D}^{k,j} v_{h} \big)_h
	+
	\big( \widetilde{D}^{k,j}_B u_{h},\,  v_{h} \big)_h
	= 0.
\end{equation*}
However, it is important to point out that the one-step difference quotient $\widetilde{D}^{k, j}$  has opposite directions on adjacent lattice nodes, which makes the weak convergence process fail. To fix this problem, we introduce multiplier
	\begin{equation*}
	\kappa(\bx_h) =
	\left\{ \begin{aligned}
		1, &\quad \mbox{if $\bx_h \in \mathcal{G}_1$},\\
		-1, &\quad \mbox{if $\bx_h \in \mathcal{G}_2$},
	\end{aligned} \right.
\end{equation*}
such that $\kappa\widetilde{D}^{k,j} u_h $ retains the uniform direction on the lattice.

When $k=2$, note that $\kappa(\boldsymbol{x}_{h}) = \kappa(\boldsymbol{x}_{h}^{k,j})$, we can deduce directly the following formula of integrating by parts. 
\begin{proposition} For every function $u_{h}(\boldsymbol{x}_{h})$, $v_{h}(\boldsymbol{x}_{h})$ defined on $\mathcal{G}$, it holds that 
\begin{equation}\label{integration by parts for k=2}
	\big(u_{h},\, \kappa\widetilde{D}^{2,j} v_{h} \big)_h
	+
	\big( \kappa\widetilde{D}^{2,j}_B u_{h},\,  v_{h} \big)_h
	= 0.
\end{equation}
\end{proposition}

	\subsection{Multi-step difference quotients}
For the case $k=1$ or $3$, the operator $\kappa\widetilde{D}^{k,j}$ does not satisfy the integration by parts property \eqref{integration by parts for k=2}.
To overcome this problem, we introduce multi-step difference quotients for the $1$st and $3$rd nearest neighbors. We denote the two-step difference quotient $\widetilde{\mathbb{D}}^{k}_{x}$ for $k=1,3$, which is given by
{\begin{equation}\label{define two step difference}
	\widetilde{\mathbb{D}}^{k}_{x}u_h(\bx_h)
	= \frac{\sqrt{3}}{3}
	\big\{\kappa\widetilde{D}^{k,3} - \kappa\widetilde{D}^{k,2}\big\} u_h(\bx_h),
\end{equation}
and $\widetilde{\mathbb{D}}^{k}_{y}$ denotes the three-step difference quotient in the $k$th nearest neighbors that is constructed as
\begin{equation}\label{three step difference quotion}
	\widetilde{\mathbb{D}}^{k}_{y} u_h(\bx_h)
	=
	\frac{1}{3} \big\{ 2\kappa \widetilde{D}^{k,1} - \kappa \widetilde{D}^{k,2} - \kappa \widetilde{D}^{k,3} \big\}
	u_h(\bx_h).
\end{equation}
We will use the operators $\widetilde{\mathbb{D}}^{k}_{x}$ and $\widetilde{\mathbb{D}}^{k}_{y}$ to conduct a weak convergence procedure when $k$ equals 1 or 3. In the following we prove that operators $\widetilde{\mathbb{D}}^{k}_{x}$ and $\widetilde{\mathbb{D}}^{k}_{y}$
yield the formula of integration by parts. 
	\begin{proposition}\label{integration by parts} When $k = 1, 3$, for every $u_{h}(\boldsymbol{x}_{h})$, $v_{h}(\boldsymbol{x}_{h})$ defined on $\mathcal{G}$, it  holds that
		\begin{equation}\label{inte by parts for two steps}
			\begin{gathered}
			\big(u_{h},\, \widetilde{\mathbb{D}}^{k}_{x} v_{h} \big)_h
			+
			\big( \widetilde{\mathbb{D}}^{k}_{x} u_{h},\,  v_{h} \big)_h
			= 0,
			\end{gathered}
		\end{equation}
	and
	\begin{equation}\label{inte by parts for three steps}
		\begin{gathered}
			\big(u_{h},\, \widetilde{\mathbb{D}}^{k}_{y} v_{h} \big)_h
			+
			\big( \widetilde{\mathbb{D}}^{k}_{y} u_{h},\,  v_{h} \big)_h = 0.
		\end{gathered}
	\end{equation}

	\end{proposition}
\begin{proof}
Multiplying  the equation  \eqref{inte by parts} by $\kappa(\boldsymbol{x}_{h})$ yields
\begin{equation*}
	\begin{aligned}
		u_{h}(\boldsymbol{x}_{h}) \cdot \kappa(\boldsymbol{x}_{h}) \widetilde{D}^{k,j} v_{h}(\boldsymbol{x}_{h})
		+
		\kappa(\boldsymbol{x}_{h}^{k,j}) \widetilde{D}^{k,j} u_{h}(\boldsymbol{x}_{h}^{k,j}) \cdot v_{h}(\boldsymbol{x}_{h}^{k,j})\\
		=
		\kappa(\boldsymbol{x}_{h}) \widetilde{D}^{k,j} \big\{ u_{h}(\boldsymbol{x}_{h}) \cdot v_{h}(\boldsymbol{x}_{h}) \big\}.
	\end{aligned}
\end{equation*}
Applying the definition in \eqref{define two step difference}, one can deduce from above equation
\begin{equation*}\label{inte by parts 2}
	\begin{aligned}
		u_{h}(\boldsymbol{x}_{h}) \cdot \widetilde{\mathbb{D}}^{k}_{x} v_{h}(\boldsymbol{x}_{h})
		+
		\widetilde{\mathbb{D}}^{k}_{x} u_{h}(\boldsymbol{x}_{h}^{k,j}) \cdot v_{h}(\boldsymbol{x}_{h}^{k,j})
		=
		\widetilde{\mathbb{D}}^{k}_{x} \big\{ u_{h}(\boldsymbol{x}_{h}) \cdot v_{h}(\boldsymbol{x}_{h}) \big\},
	\end{aligned}
\end{equation*}
By summing \eqref{inte by parts 2} over all $\bx_h\in \mathcal{G}$, we obtain the equation \eqref{inte by parts for two steps}, which follows from the fact that $\sum_{\bx_h \in \mathcal{G}} \widetilde{\mathbb{D}}^{k}_{x} f_h(\bx_h) = 0$.
The equation \eqref{inte by parts for three steps} can be proved in the same  way, for which we omit the details.
\end{proof}

	\subsection{Auxiliary difference quotients}\label{sec: Properties of hexagonal lattice}
 Finally, it remains to establish the connection between discrete Laplace operator \eqref{Green fomular} and the multi-step difference quotients \eqref{define two step difference}-\eqref{three step difference quotion}.
 To achieve this, we need  some auxiliary difference quotients. For $k=1,3$ and $r=0,1$,  define
\begin{equation}\label{define three step difference}
	\widetilde{\mathfrak{D}}^{k}_{r} u_h(\bx_h)
	=
	\left\{\begin{aligned}
		\big(\frac{2}{3}\big)^r \widetilde{D}^{k,1} u_h(\bx_h), &\quad \mbox{for $\bx_h\in \mathcal{G}_1$}\\
		\big(\frac{1}{3}\big)^r \big\{\widetilde{D}^{k,2} + \widetilde{D}^{k,3} \big\} u_h(\bx_h), &\quad \mbox{for $\bx_h\in \mathcal{G}_2$},
	\end{aligned}\right.
\end{equation}
and its backward difference quotient is given by
\begin{equation}\label{define three step backward difference}
	\widetilde{\mathfrak{D}}^{k}_{r,B} u_h(\bx_h)
	=
	\left\{\begin{aligned}
		- \big(\frac{1}{3}\big)^r \big\{\widetilde{D}^{k,2} + \widetilde{D}^{k,3} \big\} u_h(\bx_h), &\quad \mbox{for $\bx_h\in \mathcal{G}_1$},\\
		- \big(\frac{2}{3}\big)^r \widetilde{D}^{k,1} u_h(\bx_h), &\quad \mbox{for $\bx_h\in \mathcal{G}_2$}.
	\end{aligned}\right.
\end{equation}
By the definitions of \eqref{three step difference quotion}-\eqref{define three step backward difference}, one can directly deduce that
\begin{equation}\label{formula of three step difference quotion}
		\widetilde{\mathbb{D}}^{k}_{y} u_h(\bx_h)
		=
		\big\{ \widetilde{\mathfrak{D}}^{k}_{1} + \widetilde{\mathfrak{D}}^{k}_{1,B} \big\}
		u_h(\bx_h),
	\end{equation}
 Moreover, we have the following decomposition of the quadratic form of one-step operator $\widetilde{D}^{k,j}$:

\begin{lemma}
	For any $u_h$, $v_h$ defined on $\mathcal{G}$ and $k=1,3$, the following decomposition holds:
\begin{equation}\label{transform from one step to multi step}
	\begin{aligned}
		\sum_{j=1}^{3} & \widetilde{D}^{k,j} u_h \cdot \widetilde{D}^{k,j} v_h\\
		&= \frac{3}{2}
		\big\{\widetilde{\mathfrak{D}}^{k}_{1} u_h \cdot \widetilde{\mathfrak{D}}^{k}_{0} v_h
		+
		\widetilde{\mathfrak{D}}^{k}_{1,B} u_h \cdot \widetilde{\mathfrak{D}}^{k}_{0,B} v_h
		+
		\widetilde{\mathbb{D}}^{k}_{x} u_h \cdot \widetilde{\mathbb{D}}^{k}_{x} v_h\big\}.
	\end{aligned}
\end{equation}
	\end{lemma}
\begin{proof}
Note that for any $u_h$, $v_h$ defined on $\mathcal{G}$, one has
\begin{equation*}\label{transform from one step to multi step 1}
	\begin{aligned}
		\sum_{j=1}^{3}  \widetilde{D}^{k,j} u_h \cdot \widetilde{D}^{k,j} v_h
		= &
		\widetilde{D}^{k,1} u_h \cdot \widetilde{D}^{k,1} v_h
		+
		\frac{1}{2}\big\{\widetilde{D}^{k,2} + \widetilde{D}^{k,3} \big\}u_h \cdot  \big\{\widetilde{D}^{k,2} + \widetilde{D}^{k,3} \big\}v_h\\
		&+
		\frac{1}{2} \big\{\widetilde{D}^{k,2} - \widetilde{D}^{k,3} \big\}u_h \cdot  \big\{\widetilde{D}^{k,2} - \widetilde{D}^{k,3} \big\}v_h, 
	\end{aligned}
\end{equation*}
then \eqref{transform from one step to multi step} follows by substituting the definitions \eqref{define two step difference}-\eqref{define three step backward difference} into the above equation.
	
	\end{proof}

The formula \eqref{transform from one step to multi step} gives us the connection between the discrete Laplace operator and the multi-step difference quotients $\widetilde{\mathbb{D}}^{k}_{x}$, $\widetilde{\mathbb{D}}^{k}_{y}$, by the help of \eqref{Green fomular} and \eqref{formula of three step difference quotion}.
We remark that,
the auxiliary operators $\widetilde{\mathfrak{D}}^{k}_{r}$ and $\widetilde{\mathfrak{D}}^{k}_{r,B}$ cannot be used to conduct weak convergence directly, since they neither have uniform direction on the lattice when $r=0$, nor yield the integration by parts when $r=1$.

With these notations, we conclude the following convergence result:
\begin{lemma}\label{lemma: convergence of discrete difference}
	Assume that  $q_h$ is the interpolation operator given by Definition \ref{def: q_h}. For $\mathbf{\varphi}\in C_c^2(\mathbb{R}^2)$, let $\mathbf{\varphi}_h(\bx_h) = \mathbf{\varphi}(\bx_h)$. When $k=1,3$, it holds that
	\begin{equation}\label{convergence of test function}
		\begin{gathered}
			q_{h} \big\{\widetilde{\mathbb{D}}^{k}_{x}\boldsymbol{\varphi}_h\big\}
			\rightarrow (-1)^i
			\frac{\partial}{\partial x} \boldsymbol{\varphi}, \ \ \ \ \mbox{in}\ L^2\\
			q_h \big\{\widetilde{\mathfrak{D}}^{k}_{0} \boldsymbol{\varphi}_h\big\},\,\,
			q_h \big\{\widetilde{\mathfrak{D}}^{k}_{0,B}\boldsymbol{\varphi}_h\big\},\,\,
			q_h \big\{\widetilde{\mathbb{D}}^{k}_{y} \boldsymbol{\varphi}_h\big\}
			\rightarrow (-1)^i
			\frac{\partial}{\partial y} \boldsymbol{\varphi} \ \ \ \ \mbox{in}\ L^2
		\end{gathered}
	\end{equation}
with $i=0$ for $k=1$, and $i=1$ for $k=3$. When $k=2$, it holds that
\begin{equation}\label{notation of e^k,j}
	q_h \big\{\kappa \widetilde{D}^{2,j}\boldsymbol{\varphi}_h\big\},\,\,
	q_h \big\{\kappa \widetilde{D}_B^{2,j}\boldsymbol{\varphi}_h\big\}
	\rightarrow
	\frac{\partial}{\partial \boldsymbol{\nu}^{j}} \boldsymbol{\varphi} \ \ \mbox{in}\ L^2,\quad j=1,\cdots,6,
\end{equation}
here $\boldsymbol{\nu}^{j}$ denotes the fixed unit vector independent of $\bx$, given by $\boldsymbol{\nu}^{1} = - \boldsymbol{\nu}^{4} = (-1/2,\,\sqrt{3}/2)$, $\boldsymbol{\nu}^{2} = - \boldsymbol{\nu}^{5} = (-1,\,0)$, $\boldsymbol{\nu}^{3} = - \boldsymbol{\nu}^{6} = (-1/2, \,-\sqrt{3}/2)$.
\end{lemma}
Lemma \ref{lemma: convergence of discrete difference} follows directly from Taylor expansion and the property \eqref{property 3} of $q_h$, so we omit the details here.

\section{Proof of Theorem \ref{main theorem}}\label{sec: proof of convergence}
In this section, we will provide an overview of the proof for Theorem \ref{main theorem}.  
We recall that the discrete LLG equation reads:
\begin{equation}\label{discrete LLG equation}
		\left\{\begin{aligned}
			\frac{\d\bm_h}{\d t^*}(\boldsymbol{x}_{h})
			= &
			-\frac{1}{(h^*)^{2}}\Big\{\bm_{h}(\boldsymbol{x}_{h})
			\times\mathbf{B}_{\rm{eff},h}(\boldsymbol{x}_{h})\\
			& \qquad \qquad \quad +
			\alpha\bm_{h}(\boldsymbol{x}_{h})\times \big(\bm_{h}(\boldsymbol{x}_{h})\times\mathbf{B_{\rm{eff},h}}(\boldsymbol{x}_{h})\big) \Big\},\\
			\bm_{h}(\boldsymbol{x}_{h},0)
			=& \bm_{\rm{init},h}(\boldsymbol{x}_{h}),\quad
			\left|\bm_{\rm{init},h}(\boldsymbol{x}_{h}) \right| =1.
		\end{aligned}\right.
	\end{equation}
By virtue  of Cauchy-Lipshitz-Picard Theorem\cite{brezis_functional_2011}, it is guaranteed that there exists a unique solution $\bm_h \in C^{1}\left([0,+\infty)\right)$ for the ODE system \eqref{discrete LLG equation}.
Let us revisit the discrete energy \eqref{discrete energy} and provide the expression for the total effective field in \eqref{discrete effect field}. We write $\mathbf{B_{\rm{eff},h}}(\boldsymbol{x}_{h}) = \mathbf{B}_1(\boldsymbol{x}_{h}) + \mathbf{B}_2(\boldsymbol{x}_{h}) + \mathbf{B}_3(\boldsymbol{x}_{h}) + \mathbf{B}_4(\boldsymbol{x}_{h})$, 
where
\begin{equation*}
    \begin{aligned}
    \mathbf{\mathbf{B}_1(\boldsymbol{x}_{h})}
    =-
    \frac{(h^*)^2}{h^2} \frac{\delta \mathcal{H}_{\mathrm{e}}}{\delta \bm_{h}},&\qquad&
     &\mathbf{\mathbf{B}_2(\boldsymbol{x}_{h})}
    =-
    \frac{(h^*)^2}{h^2} \frac{\delta \mathcal{H}_{\mathrm{b}}}{\delta \bm_{h}},\\
    \mathbf{\mathbf{B}_3(\boldsymbol{x}_{h})}
    =-
    \frac{\delta \mathcal{H}_{\mathrm{a}}}{\delta \bm_{h}},&\qquad&
     &\mathbf{\mathbf{B}_4(\boldsymbol{x}_{h})}
    =-
    \frac{\delta \mathcal{H}_{\mathrm{z}}}{\delta \bm_{h}}.
    \end{aligned}
\end{equation*}
Then one can calculate from \eqref{discrete energy} that
\begin{equation}\label{expression of discrete effect field}
	\left\{\begin{aligned}
		\mathbf{B}_1(\boldsymbol{x}_{h})
		= &
		(h^*)^2\sum_{k=1}^{3}\sum_{j=1}^{N_{k}}\frac{J_{k}}{h^2}\bm_{h}(\boldsymbol{x}_{h}^{k,j})\\
		\mathbf{B}_2(\boldsymbol{x}_{h})
  = &
		2(h^*)^2 \sum_{j=1}^{3}\frac{K}{h^2}
		\Big(\bm_{h}(\boldsymbol{x}_{h})\cdot\bm_{h}(\boldsymbol{x}_{h}^{1,j})\Big)\bm_{h}(\boldsymbol{x}_{h}^{1,j}),\\
		\mathbf{B}_3(\boldsymbol{x}_{h})
		= &
		\sum_{k=1}^{3}\sum_{j=1}^{N_{k}}\Big(L_{k}(\boldsymbol{x}_{h} ) + L_{k}(\boldsymbol{x}_{h}^{k,j})\Big)\Big(\bm_{h}(\boldsymbol{x}_{h}^{k,j})
		\cdot\mathbf{e}^{z}\Big)\mathbf{e}^{z}\\
        & +
        2\lambda(\boldsymbol{x}_{h})\Big(\bm_{h}(\boldsymbol{x}_{h})\cdot \mathbf{e}^{x}\Big) \mathbf{e}^{x},\\
		\mathbf{B}_4(\boldsymbol{x}_{h})
		= &
		\mu\mathbf{B}(\boldsymbol{x}_{h}).
	\end{aligned}\right.
\end{equation}

\subsection{Energy estimate}
In the following, we introduce the discrete PDE type equation:
\begin{equation}\label{discrete LLG}
	\frac{\d\bm_{h}}{\d t^*}
	=
	-\bm_{h}\times\mathbf{B}^{*}_{\rm{eff},h}
	-
	\alpha\bm_{h}\times(\bm_{h}\times\mathbf{B}^{*}_{\rm{eff},h}),
\end{equation}
where the new effective field, denoted as $\mathbf{B}^{*}_{\rm{eff},h}$, is expressed using the difference quotient notation. It is given by the equation $\mathbf{B}^{*}_{\rm{eff},h}(\boldsymbol{x}_{h}) = \mathbf{B}_1^*(\boldsymbol{x}_{h}) + \mathbf{B}_2^*(\boldsymbol{x}_{h}) + \mathbf{B}_3^*(\boldsymbol{x}_{h}) + \mathbf{B}_4^*(\boldsymbol{x}_{h})$ with $\mathbf{B}_3^* = \mathbf{B}_3^{(1)} + \mathbf{B}_3^{(2)}$, which can computed by:
\begin{equation}\label{new form of discrete effect field}
	\left\{\begin{aligned}
		&\mathbf{B}_1^*(\boldsymbol{x}_{h})
		= 
		\sum_{k=1}^{3} l_{k}J_k\widetilde{\Delta}_{k}\bm_{h}(\boldsymbol{x}_{h})\\
		&\mathbf{B}_2^*(\boldsymbol{x}_{h})
  = 
  2 l_1 K \widetilde{\Delta}_{1}\bm_{h}(\boldsymbol{x}_{h})
  +
		2K\sum_{j=1}^{3}\Big(\bm_{h}(\boldsymbol{x}_{h})\cdot \widetilde{D}^{1,j}\bm_{h}(\boldsymbol{x}_{h})\Big)\widetilde{D}^{1,j}\bm_{h}(\boldsymbol{x}_{h}),\\
		&\mathbf{B}^{(1)}_3(\boldsymbol{x}_{h})
		= 
		\frac{2}{(h^*)^2} \lambda (\boldsymbol{x}_{h}) \Big(\bm_{h}(\boldsymbol{x}_{h})\cdot \mathbf{e}^{x}\Big) \mathbf{e}^{x}
		+
		\frac{2}{(h^*)^2} \sum_{k=1}^{3} N_k L_k (\boldsymbol{x}_{h}) \Big(\bm_{h}(\boldsymbol{x}_{h})\cdot\mathbf{e}^{z}\Big)
		\mathbf{e}^{z},\\
  &\mathbf{B}^{(2)}_3(\boldsymbol{x}_{h})
  =
\frac{h^2}{(h^*)^2} \sum_{k=1}^{3}
l_{k} \Big\{
L_k (\boldsymbol{x}_{h}) \Big(\widetilde{\Delta}_{k}\bm_{h}(\boldsymbol{x}_{h})\cdot\mathbf{e}^{z}\Big)\mathbf{e}^{z}
+
\Big(\widetilde{\Delta}_{k} (L_k\bm_{h})(\boldsymbol{x}_{h})\cdot\mathbf{e}^{z}\Big)\mathbf{e}^{z}\Big\}\\
		&\mathbf{B}^*_4(\boldsymbol{x}_{h})
		= 
		\frac{\mu}{(h^*)^2}\mathbf{B}(\boldsymbol{x}_{h}).
	\end{aligned}\right.
\end{equation}
The relation between the equation above and the descrete LLG equation \eqref{discrete LLG equation}-\eqref{expression of discrete effect field} is concluded in the following proposition.
\begin{proposition}\label{thm: discrete pde equation}
	Let $\bm_{h}$  be the solution of equation \eqref{discrete LLG equation} with effective field $\mathbf{B_{\rm{eff},h}}$ given by \eqref{expression of discrete effect field}. Then $\bm_{h}$ also satisfies the discrete equation \eqref{discrete LLG},
	with a PDE type effective field $\mathbf{B}^{*}_{\rm{eff},h}(\boldsymbol{x}_{h})$ given by \eqref{new form of discrete effect field}.
\end{proposition}
Moreover, the solution $\bm_h$ has the following energy estimate.
\begin{proposition}\label{thm: uniform estimate} 
Suppose the Assumption \eqref{assume} holds. If $\bm_{h}$ is the solution of the equations \eqref{discrete LLG}-\eqref{new form of discrete effect field},  it yields
\begin{equation}\label{bound of gradient}
	|\bm_{h}(\boldsymbol{x}_{h},t)|=1,
 \quad 
 \sup_{t\geq0}\|\widetilde{\nabla}_k\bm_{h}(\cdot,t)\|_{\widetilde{L}_{h}^{2}}^2\leq C,
\end{equation}
Furthermore, given $T>0$, it satisfies the following estimates 
	\begin{equation}\label{energy estimate}
	 \int_0^T \| \bm_{h}\times\mathbf{B}^*_{\rm{eff},h}\|_{\widetilde{L}_{h}^{2}}^{2} \d t \le C,
	 \quad
	 \int_{0}^{T} \|\frac{\partial\bm_{h}}{\partial t}\|_{\widetilde{L}_{h}^{2}}^2
	 \d t \leq C,
	\end{equation}
	where $C$ is a constant that depends on the initial data, $\Vert \mathbf{B} \Vert_{L^1}$, $\|L\|_{L^1}$, and $\|\lambda\|_{L^1}$, but is independent of $h$.
\end{proposition}

The Proposition \ref{thm: discrete pde equation} and \ref{thm: uniform estimate} will be verified in the Appendix. 


	\subsection{Generalization of Ladysenskaya's interpolation operators}\label{sec: beyond Taylor}
	We now introduce interpolation operators, denoted as $q_{h}$ and $p_{h}$, which are defined in Definitions \ref{def: q_h} and \ref{def: p_h}. The construction of these operators is explained in Section \ref{sec:interplolation to Hexagonal Lattice}. In order to prove Theorem \ref{main theorem}, it is necessary for the operators $q_{h}$ and $p_{h}$ to satisfy certain properties.

\begin{proposition}\label{define qh}
	For any $u_{h}$, $v_h\in \widetilde{L}_{h}^{2}$, the linear interpolation $q_h$ defined in Definition \ref{def: q_h} satisfies
	\begin{enumerate}[(i)]
		\item\label{property 1} $($isometric mapping$)$ $\int_{\mathbb{R}^2} q_h u_h \cdot q_h v_h \d \bx = (u_h,\, v_h)_h$,
		\item\label{property 2} $q_h \{u_h \cdot v_h\} = q_h u_h \cdot q_h v_h$,
		\item\label{property 3} $\inf\limits_{\bx_h\in V_E} u_h(\bx_h) \le q_h u_h (\bx) \le \sup\limits_{\bx_h\in V_E} u_h(\bx_h)$, for any $\bx\in E$.
	\end{enumerate}
	Here $E$ represents a unit cell in the hexagonal lattice, and $V_E$ is its vertex point.
\end{proposition}
\begin{proposition}\label{define ph}
	For any $u_h\in \widetilde{L}_{h}^{2}$, $v_{h}\in \widetilde{H}_{h}^{1}$, the linear interpolation $p_h$ defined in Definition \ref{def: p_h} satisfies
	\begin{enumerate}[(i)]
		\setcounter{enumi}{3}
		\item\label{property 4} $\|p_{h}u_{h}\|_{L^{2}}\leq C\|u_{h}\|_{\widetilde{L}_{h}^{2}}$,
		\quad
		$\|\nabla p_{h}v_{h}\|_{L^{2}}\leq C\|\widetilde{\nabla}_{1} v_{h}\|_{\widetilde{L}_{h}^{2}}$,
		\item\label{property 5} 
		$
			\|p_{h}v_{h} - q_{h}v_{h}\|_{L^2}
			\leq
			h \|\widetilde{\nabla}_{1} v_{h}\|_{\widetilde{L}_{h}^{2}},
		$\par
	\end{enumerate}
	where the constant $C$ is independent of $u_h$, $v_h$ and $h$.
\end{proposition}

For brevity, the proofs of Propositions \ref{define qh} and \ref{define ph} are left in Section \ref{subsec:Proof of Propositions of qh and ph}.

\begin{remark}
	Here, the required properties are provided to support the weak convergence argument. 
\end{remark}

\begin{remark}\label{remark: introduction of rh}
	Similar interpolation operators $p_h$ and $q_h$ were introduced in \cite{ladyzhenskaya_boundary_1985} for the square lattice case. For any  discrete function $u_h$, the function $q_hu_h$ is constructed as a step function in $\mathbb{R}^2$, and $p_hu_h$ is constructed as a piece-wise linear function in $\mathbb{R}^2$.
	
	Besides $q_h$ and $p_h$, one more interpolation operator $r_h$ was suggested in \cite{ladyzhenskaya_boundary_1985} to conduct the weak convergence of derivatives. More precisely, denote a differential operator $D$ and its discrete form $\widetilde{D}$, and given a convergent sequence
	\begin{equation*}\label{convergence sequence}
		q_hu_h,\, p_hu_h \rightarrow u \mbox{ strongly,}\quad\ \ \
		q_h \widetilde{D} u_h \rightarrow v \mbox{ weakly,}
	\end{equation*}
	in $L^2$. $r_h$ was introduced to  prove $v = D u$.
	
Defining an interpolation operator $r_h$ on a hexagonal lattice poses difficulties. Instead, we propose an alternative strategy to prove $v=Du$ in this work. This approach will be applied in the upcoming proof of \eqref{weak convergence of derivatives in space for k=2} and \eqref{weak convergence of derivatives in space}.
\end{remark}

\subsection{Convergence of $\bm_h$ and its derivatives}
\label{sec:existence}
  In this subsection, we will prove the convergence of the discrete solution $\bm_h$ and its derivatives, with the aid of the uniform estimates in Proposition \ref{thm: uniform estimate} and interpolation operator $q_h$ and $p_h$ which satisfy Propositions \ref{define qh} and \ref{define ph}. 
  
%
  
   Assume that $q_h$ and $p_h$ are interpolation operators satisfying   Propositions \ref{define qh} and \ref{define ph}.
Applying properties \eqref{property 1} and \eqref{property 4}, one can deduce from \eqref{bound of gradient}-\eqref{energy estimate}
\begin{equation}\label{bounded result}
	\begin{gathered}
		p_{h}\bm_{h}\  \mbox{remains in a bounded set of  $L^{\infty}(0,T;\dot{H}^{1})\cap L^\infty(0,T;L^\infty)$},\\
		\frac{d}{dt}p_{h}\bm_{h}, \,\frac{d}{dt}q_{h}\bm_{h}
		\mbox{ remain in a bounded set of $L^{2}(0,T;L^{2})$}.
	\end{gathered}
\end{equation}
As a consequence of Aubin-Lions Lemma, there exists a subsequence of $p_{h}\bm_{h}$ converging strongly to some  function $\bm^*$ in $C([0,T];L^{2}_{\mathrm{loc}})$. Due to the  property \eqref{property 5} and \eqref{bound of gradient},  it as well holds the convergence
\begin{gather}\label{strong convergence of m}
	q_{h}\bm_{h} \rightarrow \bm^*\quad  \mathrm{strongly\ in}\ C([0,T];L^{2}_{\mathrm{loc}}),
\end{gather}
and thus converges almost everywhere. Note that $\abs{q_{h}\bm_{h}(\bx)} = 1$ for all $\bx \in \mathbb{R}^2$, it follows from the properties \eqref{property 2}-\eqref{property 3} that 
\begin{equation}\label{|m|=1}
\abs{\bm^*} = 1 \quad \text{a.e in $\mathbb{R}^2$}.
\end{equation}
 Moreover, the statement \eqref{bounded result} tells that a subsequence of $q_h\bm_h$ (which we also denote by $q_h\bm_h$) satisfies
\begin{gather}\label{weak convergence of derivatives in time}
	\frac{\partial}{\partial t}q_{h}\bm_{h} \rightarrow \frac{\partial\bm^*}{\partial t}
	\quad
	\mathrm{weakly\ in}\ L^{2}(0,T;L^{2}).
\end{gather}

Now we consider the convergence of spatial derivatives. When $k=2$, according to the  estimate \eqref{bound of gradient},  $q_{h} \big\{\kappa \widetilde{D}^{k,j} \bm_{h}\big\}$ remains in a bounded set of $L^{\infty}(0,T,L^{2})$. Thus it has a subsequence which converges to some limit function weakly * in $L^{\infty}(0,T,L^{2})$. We assert that, this limit function is $\frac{\partial \bm^*}{\partial \nu^j}$. In fact, due to the formula of integration by parts and the property \eqref{property 1} for $q_h$, it holds that 
\begin{equation*}
	\begin{aligned}
		\int_{\mathbb{R}^{2}}q_{h} \big\{\kappa\widetilde{D}^{2,j} \bm_{h} \big\}
		\cdot q_h \boldsymbol{\mathbf{\varphi}_h} \mathrm{d}\boldsymbol{x}
		= -
		\int_{\mathbb{R}^{2}}q_{h}\bm_{h}
		\cdot q_h \big\{\kappa\widetilde{D}^{2,j}_B \boldsymbol{\mathbf{\varphi}}_h\big\} \mathrm{d}\boldsymbol{x}, 
	\end{aligned}
\end{equation*}
here, $\boldsymbol{\mathbf{\varphi}}_h$ is the discrete function satisfying $\boldsymbol{\mathbf{\varphi}}_h(\boldsymbol{x}_h) = \boldsymbol{\mathbf{\varphi}}(\boldsymbol{x}_h)$ for every
$\boldsymbol{\mathbf{\varphi}}(x) \in C_c^\infty(\mathbb{R}^2)$. By virtue of Lemma \ref{lemma: convergence of discrete difference} and  \eqref{strong convergence of m}, one has  
\begin{equation}\label{m_h weak convergence}
\int_{\mathbb{R}^2} q_h \bm_h \cdot q_h \left\{\kappa \widetilde{D}^{2,j}_B \boldsymbol{\mathbf{\varphi}}_h \right\}\mathrm{d}\boldsymbol{x} \rightarrow 
\int_{\mathbb{R}^2} \bm^* \cdot \frac{\partial \boldsymbol{\mathbf{\varphi}}}{\partial \nu_j}\mathrm{d}\boldsymbol{x}.
\end{equation}
On the other hand
\begin{equation*}
	\begin{aligned}
		\Big|\int_{\mathbb{R}^{2}}&q_{h} \big\{\kappa\widetilde{D}^{2,j} \bm_{h} \big\}
		\cdot \boldsymbol{\mathbf{\varphi}} \mathrm{d}\boldsymbol{x}
		-\int_{\mathbb{R}^{2}}q_{h} \big\{\kappa\widetilde{D}^{2,j} \bm_{h} \big\}
		\cdot q_h \boldsymbol{\mathbf{\varphi}_h} \mathrm{d}\boldsymbol{x} \Big|\\
		=&
		\Big|\int_{\mathbb{R}^{2}}q_{h} \big\{\kappa\widetilde{D}^{2,j} \bm_{h} \big\}
		\cdot (\boldsymbol{\mathbf{\varphi}}-q_h \boldsymbol{\mathbf{\varphi}_h}) \mathrm{d}\boldsymbol{x}
		\Big|
		\leq \big\|q_{h} \big\{\kappa\widetilde{D}^{2,j} \bm_{h} \big\}\big\|_{L^2}\|\boldsymbol{\mathbf{\varphi}}-q_h \boldsymbol{\mathbf{\varphi}_h}\|_{L^2},
		\end{aligned}
	\end{equation*}
which together with \eqref{m_h weak convergence} implies that
\begin{equation}\label{weak convergence of derivatives in space for k=2}
	\begin{gathered}
		q_{h} \big\{\kappa \widetilde{D}^{2,j}\bm_{h}\big\}
		\rightarrow
		\frac{\partial \bm^*}{\partial \nu^j}
		\quad
		\mbox{weakly in $L^{\infty}(0,T,L^{2})$},
	\end{gathered}
\end{equation}
for all $j=1,\cdots,6$. 

As for the case $k=1,3$,
note that the sequence $\widetilde{\mathbb{D}}^{k}_{x}\bm_{h}$ and
$ \widetilde{\mathbb{D}}^{k}_{y} \bm_{h} $
are also bounded in $L^{\infty}(0,T,\widetilde{L}^{2}_h)$ by \eqref{bound of gradient}, one can repeat the above step and use the result of integration by parts in Lemma \ref{integration by parts} and strong convergence in Lemma \ref{lemma: convergence of discrete difference}, to derive that
\begin{equation}\label{weak convergence of derivatives in space}
\begin{gathered}
	q_{h} \big\{\widetilde{\mathbb{D}}^{k}_{x}\bm_{h}\big\}
	\rightarrow
	\frac{\partial \bm^*}{\partial x} ,
	\quad
	q_{h} \big\{ \widetilde{\mathbb{D}}^{k}_{y} \bm_{h} \big\}
	\rightarrow
	\frac{\partial \bm^*}{\partial y} ,
\end{gathered}
\end{equation}
weakly in $L^{\infty}(0,T,L^{2})$.

\subsection{Convergence of degenerated Laplace term}
Let us consider the limit for the integral of degenerated Laplace term
\begin{equation*}
	T_h := \int_{\mathbb{R}^{2}}q_{h} \{ \bm_{h} \times \widetilde{\Delta}_{k}\bm_{h} \}
	\cdot q_h \boldsymbol{\mathbf{\varphi}}_h \mathrm{d}\boldsymbol{x}. 
\end{equation*}
According to Lemma \ref{integration by parts for one step} and the property \eqref{property 1} of $q_h$, 
\begin{align}\label{degenerated Laplace term}
	T_h
	= - \frac{d_k^2}{2l_k} \sum_{j=1}^{N_{k}}
	\int_{\mathbb{R}^{2}}q_{h} \big\{ \bm_{h} \times \widetilde{D}^{k,j} \bm_{h} \big\}
	\cdot q_h \big\{\widetilde{D}^{k,j} \boldsymbol{\mathbf{\varphi}}_h\big\} \mathrm{d}\boldsymbol{x}.
\end{align}
Now we assert that 
\begin{equation}\label{convergence of m times Delta m}
	\begin{aligned}
		&\lim\limits_{h\rightarrow 0}  T_h
		= -
		\int_{\mathbb{R}^{2}} \bm^* \times \nabla\bm^*
		\cdot \nabla \boldsymbol{\mathbf{\varphi}} \mathrm{d}\boldsymbol{x}
		=
		\int_{\mathbb{R}^{2}} \big\{\nabla \cdot (\bm^* \times \nabla\bm^*)\big\}
		\cdot \boldsymbol{\mathbf{\varphi}} \mathrm{d}\boldsymbol{x}.
	\end{aligned}
\end{equation}
\begin{proof}[\textbf{Proof of \eqref{convergence of m times Delta m}}]
The proof is divided  into  two cases: $k=2$ and $k=1, 3$. 
For the case $k=2$, due to the fact $\kappa^2(\bx_h) =1$, one has 
\begin{align*}
	T_h
	=& - \frac{d_k^2}{2l_k} \sum_{j=1}^{N_{k}}
	\int_{\mathbb{R}^{2}}q_{h} \big\{ \bm_{h} \times \kappa\widetilde{D}^{k,j} \bm_{h} \big\}
	\cdot q_h \big\{\kappa\widetilde{D}^{k,j} \boldsymbol{\mathbf{\varphi}}_h\big\} \mathrm{d}\boldsymbol{x}\\
	=&- \frac{1}{3} \sum_{j=1}^{6}
	\int_{\mathbb{R}^{2}}q_{h}\bm_{h} \times q_h\big\{  \kappa\widetilde{D}^{2,j} \bm_{h} \big\}
	\cdot q_h \big\{\kappa\widetilde{D}^{2,j} \boldsymbol{\mathbf{\varphi}}_h\big\} \mathrm{d}\boldsymbol{x}.
\end{align*}
 Let $h \rightarrow 0$, it follows from \eqref{strong convergence of m}, \eqref{weak convergence of derivatives in space for k=2} and Lemma \ref{lemma: convergence of discrete difference} that  
 \begin{equation*}
 \lim\limits_{h\rightarrow 0}
		T_h =
-\frac{1}{3} \sum_{j=1}^{6}
\int_{\mathbb{R}^2} \left(\bm^* \times \frac{\partial \bm^*}{\partial \nu^j} \right) \cdot \frac{\partial \boldsymbol{\mathbf{\varphi}}}{\partial \nu^j} \mathrm{d}\boldsymbol{x} = - \int_{\mathbb{R}^2} \bm \times  \nabla \bm \cdot \nabla \boldsymbol{\mathbf{\varphi}} \mathrm{d}\boldsymbol{x},
 \end{equation*}
 with $\nu^j, 1\le j \le 6, $ given in Lemma \ref{lemma: convergence of discrete difference}.

 For the case $k=1,3$, applying the formula \eqref{transform from one step to multi step} to the right-hand side of \eqref{degenerated Laplace term}, one gets 
\begin{equation*}
	\begin{aligned}
		T_h =
		&-
		\int_{\mathbb{R}^{2}} \Big(q_{h} \{ \bm_{h} \times \widetilde{\mathfrak{D}}^{k}_{1} \bm_{h} \}
		\cdot q_h \big\{\widetilde{\mathfrak{D}}^{k}_{0} \boldsymbol{\mathbf{\varphi}}_h\big\}\Big) \\
		&+
		\Big(q_{h} \big\{ \bm_{h} \times \widetilde{\mathfrak{D}}^{k}_{1,B }\bm_{h} \big\}
		\cdot q_h \big\{\widetilde{\mathfrak{D}}^{k}_{0,B} \boldsymbol{\mathbf{\varphi}}_h\big\}\Big)
		+
		\Big(q_{h} \big\{ \bm_{h} \times \widetilde{\mathbb{D}}^{k}_{x} \bm_{h} \big\}
		\cdot q_h \big\{\widetilde{\mathbb{D}}^{k}_{x} \boldsymbol{\mathbf{\varphi}}_h\big\}\Big) \mathrm{d}\boldsymbol{x}.
	\end{aligned}
\end{equation*}
Thus using the strong convergence for $\widetilde{\mathfrak{D}}^{k}_{0}$ $\widetilde{\mathfrak{D}}^{k}_{0,B}$ by \eqref{convergence of test function},
together with formula \eqref{formula of three step difference quotion}, one can write
\begin{equation*}
	\begin{aligned}
		\lim\limits_{h\rightarrow 0}
		T_h
		=  -
		\lim\limits_{h\rightarrow 0} \int_{\mathbb{R}^{2}}
		\Big(q_{h} \big\{ \bm_{h} \times \widetilde{\mathbb{D}}^{k}_{x} \bm_{h} \big\}
		\cdot \frac{\partial \boldsymbol{\mathbf{\varphi}}}{\partial x} \Big)
		+
		 \Big(q_{h} \big\{ \bm_{h} \times \widetilde{\mathbb{D}}^{k}_{y} \bm_{h} \big\}
		\cdot \frac{\partial \boldsymbol{\mathbf{\varphi}}}{\partial y} \Big) \mathrm{d}\boldsymbol{x}.
	\end{aligned}
\end{equation*}
Using convergence result \eqref{strong convergence of m}, \eqref{weak convergence of derivatives in space}, and property \eqref{property 2}  of $q_h$, it finally leads to \eqref{convergence of m times Delta m}, for $k=1,3$.
\end{proof}

\subsection{Convergence of equation}
For every $\boldsymbol{\mathbf{\varphi}}\in C_{c}^{2}$,  define $\boldsymbol{\mathbf{\varphi}}_h(\bx_h) = \boldsymbol{\varphi}(\bx_h)$. Applying interpolation operator $q_h$ to the equation \eqref{discrete LLG}, and taking $q_{h} \boldsymbol{\mathbf{\varphi}}_h$ as a test function, one gets
\begin{align}\label{integral form}
	\begin{split}
		\int_{\mathbb{R}^{2}} \frac{\partial}{\partial t}q_{h}\bm_{h} \cdot q_{h} \boldsymbol{\mathbf{\varphi}}_{h} \mathrm{d}\boldsymbol{x}
		=&
		-\int_{\mathbb{R}^{2}} q_{h}\big\{\bm_{h}\times\mathbf{B}_{\rm {eff},h}^{*}\big\} \cdot q_{h}\boldsymbol{\mathbf{\varphi}}_{h}\mathrm{d}\boldsymbol{x}\\
		& -
		\alpha \int_{\mathbb{R}^{2}} q_{h}\big\{\bm_{h}\times
		(\bm_{h}\times\mathbf{B}_{\rm {eff},h}^{*})\big\}
		\cdot q_{h}\boldsymbol{\mathbf{\varphi}}_{h}\mathrm{d}\boldsymbol{x}.\\
	\end{split}
\end{align}
    First, we consider the limit of the first term on right-hand side and prove that
\begin{align}\label{convergence of precession term}
	\int_{\mathbb{R}^{2}} q_{h}(\bm_{h}\times\mathbf{B}_{\rm {eff},h}^{*})\cdot q_{h}\boldsymbol{\mathbf{\varphi}}_{h}\mathrm{d}\boldsymbol{x}
	\rightarrow
	\int_{\mathbb{R}^{2}} \bm^* \times\mathbf{B}_{\rm {eff}}\cdot \boldsymbol{\mathbf{\varphi}} \mathrm{d}\boldsymbol{x}.
\end{align}
 Substitute the expression of $\mathbf{B}_{\rm {eff},h}^{*} = \mathbf{B}_1^{*} + \mathbf{B}_2^{*} + \mathbf{B}_3^{*} + \mathbf{B}_4^{*}$ in \eqref{new form of discrete effect field} to the left-hand side of equation \eqref{convergence of precession term}. Then one can easily pass to the limit in \eqref{convergence of precession term} by utilizing \eqref{strong convergence of m} and \eqref{convergence of m times Delta m} with two exemptions:  
the second term in $\mathbf{B}_2^*$, and the term $\mathbf{B}_3^{(2)}$. For the second term in $\mathbf{B}_2^*$, let us denote it by
\begin{equation*}
		\widetilde{\mathbf{B}}_2^* =
		2K \sum_{j=1}^{3} \int_{\mathbb{R}^{2}} q_{h}\big\{
		(\bm_{h}\cdot \widetilde{D}^{1,j} \bm_{h})
		\bm_{h}\times \widetilde{D}^{1,j}\bm_{h}\big\}
		\cdot q_{h}\mathbf{\boldsymbol{\varphi}}_{h}\mathrm{d}\boldsymbol{x},
\end{equation*}
Since $|\bm_h(\bx_h)| = 1$ for all $\bx_h$, it follows from the definition of $\widetilde{D}^{1, j}$ that
\begin{equation*}
	\begin{aligned}
	\bm_{h}(\boldsymbol{x}_h)\cdot\widetilde{D}^{1,j}\bm_{h}(\boldsymbol{x}_{h})
	=&
	\bm_{h}(\boldsymbol{x}_h^{1,j})\cdot\widetilde{D}^{1,j}\bm_{h}(\boldsymbol{x}_{h}^{1,j}), \\
	\bm_{h}(\boldsymbol{x}_h)\times\widetilde{D}^{1,j}\bm_{h}(\boldsymbol{x}_{h})
	=&
	-\bm_{h}(\boldsymbol{x}_h^{1,j})\times\widetilde{D}^{1,j}\bm_{h}(\boldsymbol{x}_{h}^{1,j}).
	\end{aligned}
	\end{equation*}
Consequently,
\begin{equation}\label{term 1 in K part}
	\begin{aligned}
		&\big(\bm_{h}(\boldsymbol{x}_{h})
		\cdot \widetilde{D}^{1,j}\bm_{h}(\boldsymbol{x}_{h})\big)
		\bm_{h}(\boldsymbol{x}_{h})\times \widetilde{D}^{1,j}\bm_{h}(\boldsymbol{x}_{h})\\
		=&\frac{h}{2}
		\widetilde{D}^{1,j} \Big\{\big(\bm_{h}(\boldsymbol{x}_{h})
		\cdot \widetilde{D}^{1,j}\bm_{h}(\boldsymbol{x}_{h})\big)
		\bm_{h}(\boldsymbol{x}_{h})\times \widetilde{D}^{1,j}\bm_{h}(\boldsymbol{x}_{h})\Big\}.
	\end{aligned}
\end{equation}
Substituting \eqref{term 1 in K part} into $F_1$, and applying integration by parts in Lemma \ref{integration by parts} and  H\"older's inequality implies
\begin{equation*}
	\begin{aligned}
		\abs{\widetilde{\mathbf{B}}_2^*}
		\le &
		h K_1 \sum_{j=1}^{3} \|  \widetilde{D}^{1,j} \boldsymbol{\varphi}_{h} \|_{\widetilde{L}^\infty_h}
		\Vert \widetilde{D}^{1,j} \bm_{h} \Vert_{\widetilde{L}_{h}^{2}}^2,
	\end{aligned}
\end{equation*}
which tends to $0$ as  $h\rightarrow 0$. Finally, for the term $\mathbf{B}_3^{(2)}$, utilizing integration by parts yields
\begin{equation*}
     \begin{aligned}
     &\int_{\mathbb{R}^{2}} q_{h}(\bm_{h}\times \mathbf{B}^{(2)}_3)\cdot q_{h}\boldsymbol{\mathbf{\varphi}}_{h}\mathrm{d}\boldsymbol{x}\\
	= &
		\frac{h^2}{(h^*)^2} \sum_{k=1}^{3} \sum_{j=1}^{N_k} \frac{d_k^2}{2} 
  \int_{\mathbb{R}^{2}} q_h \big\{\widetilde{D}^{k,j} (L_k\bm_{h})\cdot\mathbf{e}^{z}\big\} \cdot 
  q_{h} \big\{ \widetilde{D}^{k,j} (\bm_{h}\times \mathbf{e}^{z} 
  \cdot \mathbf{\boldsymbol{\varphi}}_{h}) \big\}
  \mathrm{d}\boldsymbol{x}\\
		&+
		\frac{h^2}{(h^*)^2} \sum_{k=1}^{3} \sum_{j=1}^{N_k}\frac{d_k^2}{2} 
  \int_{\mathbb{R}^{2}} q_h \big\{ \widetilde{D}^{k,j} \bm_{h}\cdot\mathbf{e}^{z}\big\} \cdot 
  q_h \big\{
  \widetilde{D}^{k,j} (L_k\bm_{h} \times \mathbf{e}^{z} 
  \cdot \mathbf{\boldsymbol{\varphi}}_{h} ) \big\}
  \mathrm{d}\boldsymbol{x}.
     \end{aligned}
 \end{equation*}
 Therefore by H\"older's inequality one can derive
\begin{equation*}
	\begin{aligned}
		&\big|\int_{\mathbb{R}^{2}} q_{h}(\bm_{h}\times \mathbf{B}^{(2)}_3)\cdot q_{h}\boldsymbol{\mathbf{\varphi}}_{h}\mathrm{d}\boldsymbol{x} \big|\\
		\le &
		\frac{h^2}{(h^*)^2} \sum_{k=1}^{3} \sum_{j=1}^{N_k} d_k^2
		\Big(\Vert \widetilde{D}^{k,j}\bm_{h} \Vert_{\widetilde{L}_{h}^{2}}
		\Vert \widetilde{D}^{k,j} \boldsymbol{\varphi}_{h}\Vert_{\widetilde{L}_{h}^{2}}
		+
	\| \mathbf{\boldsymbol{\varphi}}_{h} \|_{\widetilde{L}^\infty_h}
		\Vert \widetilde{D}^{1,j} \bm_{h} \Vert_{\widetilde{L}_{h}^{2}}^2\Big),
	\end{aligned}
\end{equation*}
thus it tends to $0$ as  $h\rightarrow 0$. Hence the assertion \eqref{convergence of precession term} is proved.

Furthermore, inequality \eqref{bound of gradient} demonstrates that the term $q_h\big\{\bm_{h}\times\mathbf{B}_{\rm {eff},h}^{*}\big\}$ is bounded in $L^{2}(0,T,L^{2})$. Therefore there is a subsequence which converges weakly in $L^{2}(0,T,L^{2})$. Considering \eqref{convergence of precession term}, it can be inferred that
\begin{gather}\label{weak convergence of precession term}
	q_h\big\{\bm_{h}\times\mathbf{B}_{\rm {eff},h}^{*}\big\}
	\rightarrow
	\bm^* \times\mathbf{B}_{\rm {eff}}
	\quad
	\mathrm{weakly\ in}\ L^{2}(0,T,L^{2}).
\end{gather}
Let us return to the equation \eqref{integral form}. One can take the limit as $h \rightarrow 0$ for the first term in \eqref{integral form} by using \eqref{convergence of precession term}. Moreover, the last term in \eqref{integral form} can pass to the limit by utilizing \eqref{strong convergence of m} and \eqref{weak convergence of precession term}. Consequently, by taking the limit in \eqref{integral form}, we can conclude that limiting function $\bm^*$ satisfies the continuity equation \eqref{define weak solution}.

\section{Interpolation to hexagonal Lattice}
\label{sec:interplolation to Hexagonal Lattice}
In this section, we provide the construction of interpolation operators $q_h$ and $p_h$, which fulfill Propositions \ref{define qh} and \ref{define ph}.
Following the approach used in \cite{ladyzhenskaya_boundary_1985} for the square lattice case, we define the interpolation function $q_h u_h$ as a step function and $p_h u_h$ as a piece-wise linear function for any discrete function $u_h$.

\subsection{Step interpolation function}
To introduce the definition of $q_h$, we first need to establish a mapping from lattice points in $\mathcal{G}$ to open sets in $\mathbb{R}^2$. We use open set $\Omega_{\boldsymbol{x}_{h}}\subset \mathbb{R}^2$ to denote an adjoint area of $\boldsymbol{x}_{h}\in \mathcal{G}$, such that
\begin{equation*}
	\begin{aligned}
		&\mbox{if $\boldsymbol{x}_{h}\in \mathcal{G}_{1}$, then $\Omega_{\boldsymbol{x}_{h}}$ is the trapezoid on the upper right side of $\boldsymbol{x}_{h}$},\\
		&\mbox{if $\boldsymbol{x}_{h}\in \mathcal{G}_{2}$, then $\Omega_{\boldsymbol{x}_{h}}$ is the trapezoid on the lower left side of $\boldsymbol{x}_{h}$},
	\end{aligned}
\end{equation*}
as shown in Figure \ref{Fig2}.
It is easy to verify that the sets $\Omega_{\boldsymbol{x}_{h}}$ are mutually disjoint for all $\boldsymbol{x}_{h}\in \mathcal{G}$. Furthermore, the entire space of $\mathbb{R}^2$ is covered by $\mathop{\cup}_{\boldsymbol{x}_{h}\in \mathcal{G}} \bar{\Omega}_{\boldsymbol{x}_{h}}$.
\begin{figure}[htbp]
	\centering
	\includegraphics[scale=0.60]{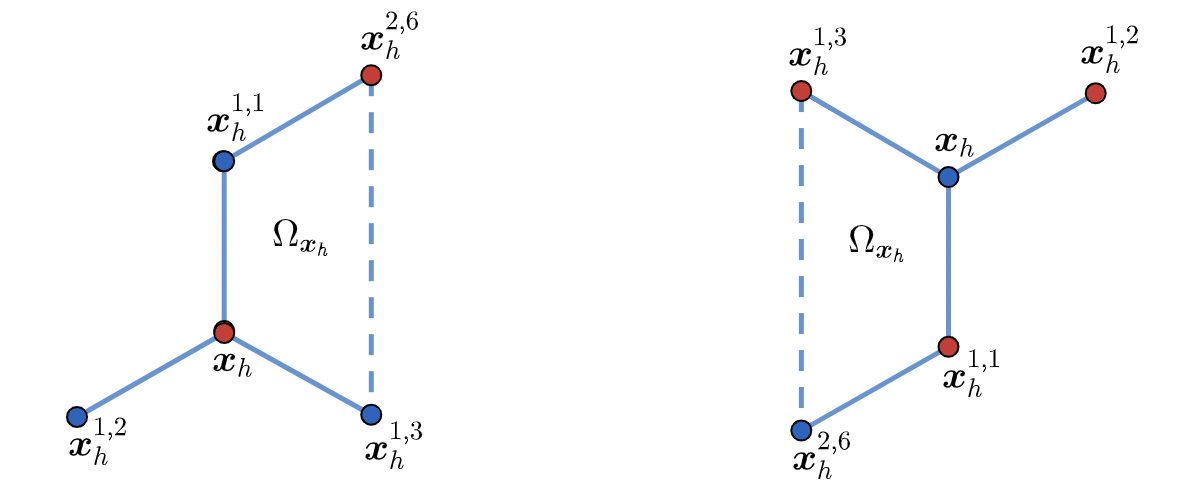}
	\caption{Lattice division 1. }
	\label{Fig2}
\end{figure}
\begin{definition}\label{def: q_h}
	Given any function $u_h(\bx_h)$ defined on $\mathcal{G}$,  define a step function $q_{h}u_{h}(\bx)$ as
	\begin{equation*}\label{definition of q_{h}}
		q_{h}u_{h}(\boldsymbol{x})=u_h(\boldsymbol{x}_{h}),\quad
		\mbox{for any $\boldsymbol{x}\in\Omega_{\boldsymbol{x}_{h}}$, $\boldsymbol{x}_{h}\in \mathcal{G}$}.
	\end{equation*}
\end{definition}

\subsection{Piece-wise linear interpolation function} The step function $q_h u_h$ is simple to evaluate, but it lacks continuity. Therefore, we require a piecewise linear function $p_h u_h$ that maintains higher regularity from $u_h$. Before defining $p_h$, let us first discuss the construction strategy.

Let $\boldsymbol{x}_{h}\in \mathcal{G}$ be fixed. We denote the lattice point $\boldsymbol{x}_{h}$ by the letter $A$, and use $B$, $C$, and $D$ to represent other vertex points of $\Omega_{\boldsymbol{x}_{h}}$ as shown in Figure \ref{Fig3}. For any function $u_h$ defined on $\mathcal{G}$, the value of $p_h u_h(\bx^*)$ for any point $\bx^*\in \Omega_{\boldsymbol{x}_{h}}$ is given by:
\begin{itemize}
	\item To begin with, we can determine the value of $u_h$ at lattice points $A$, $B$, $C$, and $D$ to evaluate $p_h u_h$. Additionally, we can use the linear Lagrangian interpolation formula to calculate $p_h u_h$ on lines $AB$ and $CD$.
	Next, we need to draw a plumb line through $\bx^*$ that intersects with $AB$ and $CD$ at points $E$ and $F$ respectively.
	\item Finally, we note that the values of $p_h u_h$ on $E$ and $F$ are given in the first step. Using linear Lagrangian interpolation formula again, we can obtain the value of $p_h u_h$ on line $EF$. Therefore, we can determine the value of $p_h u_h(\bx^*)$.
\end{itemize}

\begin{figure}[htbp]
	\centering
	\includegraphics[scale=0.5]{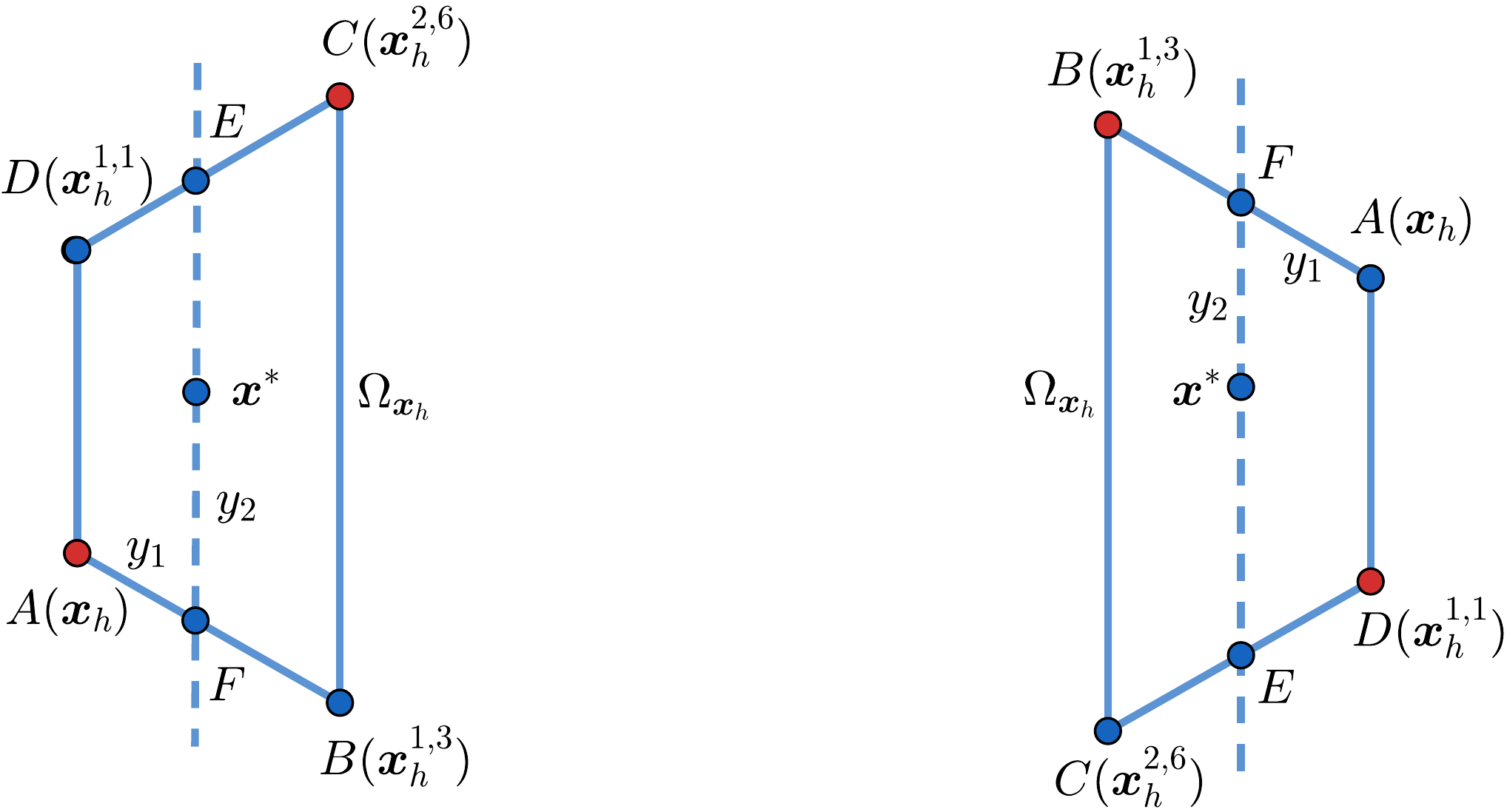}
	\caption{Lattice division 2. }
	\label{Fig3}
\end{figure}

Following the above steps, we can introduce the definition of operator $p_h$.
\begin{definition}\label{def: p_h}
	For any function $u_h$ defined on hexagonal lattice $\mathcal{G}$, define the  piece-wise linear function $p_{h}u_{h}(\boldsymbol{x}^*)$ as
	\begin{align*}\label{definition of q_{h}2}
		\begin{split}
			p_{h}u_{h}(\boldsymbol{x}^*)
			=
			&u_{h}(\boldsymbol{x}_{h})
			+ \frac{y_1}{h}
			\big(u_{h}(\boldsymbol{x}_{h}^{1,3}) - u_{h}(\boldsymbol{x}_{h})\big)
			+\frac{y_2}{y_1 + h}
			\big(u_{h}(\boldsymbol{x}_{h}^{1,1}) - u_{h}(\boldsymbol{x}_{h})\big)\\
			&+
			\frac{y_1 y_2}{h (y_1 + h)}
			\big(u_{h}(\boldsymbol{x}_{h}) - u_{h}(\boldsymbol{x}_{h}^{1,3}) - u_{h}(\boldsymbol{x}_{h}^{1,1}) + u_h (\boldsymbol{x}_{h}^{2,6})\big),
		\end{split}
	\end{align*}
for any $\boldsymbol{x}^* \in \Omega_{\boldsymbol{x}_{h}}$, $\boldsymbol{x}_{h}\in \mathcal{G}$. As shown in Figure \ref{Fig3}, the pair $(y_1, y_2)$ is determined by $\bx^*$,
\begin{align*}
	y_1 &= \frac{2}{\sqrt{3}} \abs{(\bx^* - \bx_h)\cdot \be_1},\quad
	y_2 = \frac{1}{\sqrt{3}} \abs{(\bx^* - \bx_h)\cdot \be_1}
	+
	\abs{(\bx^* - \bx_h)\cdot \be_2}.
\end{align*}

\end{definition}

\subsection{Proof of Propositions \ref{define qh} and \ref{define ph}}\label{subsec:Proof of Propositions of qh and ph}
Now let us prove the above interpolation operators $q_h$, $p_h$ constructed on hexagonal lattice  satisfy  Propositions \ref{define qh} and \ref{define ph}.

For any function $u_h$, $v_h \in \tilde{L}_h^2$, it follows from Definition \ref{def: q_h}  that
\begin{equation}\label{C5}
	\begin{split}
		\int_{R^2} q^*_h u_h \cdot q^*_h v_h \d \bx
		&=
		\sum_{\boldsymbol{x}_{h}\in \mathcal{G}} \int_{\Omega_{\boldsymbol{x}_{h}}} u_{h}(\boldsymbol{x}_{h}) \cdot v_{h}(\boldsymbol{x}_{h}) \mathrm{d}\boldsymbol{x}\\
		&=
		\frac{3\sqrt{3}}{4} h^{2}\sum_{\boldsymbol{x}_{h}\in \mathcal{G}} u_{h}(\boldsymbol{x}_{h}) \cdot v_{h}(\boldsymbol{x}_{h})
		=
		(u_h,\, v_h)_h.
	\end{split}
\end{equation}
Hence the map $q_h$ satisfies the isometry property \eqref{property 1} stated in Proposition \ref{define qh}, and it is a linear transformation from $\widetilde{L}^2_h$ to $L^2$. Additionally, properties \eqref{property 2}-\eqref{property 3} can be directly derived for $q_h$ based on its definition.

As for the operator $p_{h}$,  straightforward calculations show that
\begin{equation}\label{3}
\begin{split}
	|p_{h}u_{h}(\boldsymbol{x})|
	&\leq |u_{h}(\boldsymbol{x}_h)|
	+ |u_{h}(\boldsymbol{x}_h^{1,1})| + |u_{h}(\boldsymbol{x}_h^{1,3})|
	+ |u_{h}(\boldsymbol{x}_h^{2,6})|.
\end{split}
\end{equation}
Moreover, since $\nabla y_1 = (2/\sqrt{3},\, 0)$ and $\nabla y_2 = (1/\sqrt{3},\, 1)$, one can obtain
\begin{equation}
	|\nabla \big(p^*_{h}u_{h}(\boldsymbol{x})\big)|
	\leq
	|\widetilde{D}^{1,1}u_h(\boldsymbol{x}_{h})|
	+
	|\widetilde{D}^{1,3}u_h(\boldsymbol{x}_{h})|
	+
	|\widetilde{D}^{1,2}u_h(\boldsymbol{x}_{h}^{1,1})|.
\end{equation}
Therefore by substitution one can deduce
\begin{equation*}
\begin{split}
	\|p_{h}v_{h}\|_{L^{2}}^2
	&=
	\sum_{\boldsymbol{x}_{h}\in \mathcal{G}}
	\int_{\Omega_{\boldsymbol{x}_{h}}}|p_{h}v_{h}|^{2} \mathrm{d}\boldsymbol{x}
	\le
	\sum_{\boldsymbol{x}_{h}\in \mathcal{G}}
	\abs{\Omega_{\boldsymbol{x}_{h}}}
	\sup\limits_{\bx\in  \Omega_{\boldsymbol{x}_{h}}}
	|p_{h}u_{h}(\boldsymbol{x})|^2
	\le
	16 \|u_{h}\|_{\widetilde{L}_{h}^{2}}^2,
\end{split}
\end{equation*}
and
\begin{equation*}
\begin{split}
	\|\nabla \big(p_{h}v_{h}\big)\|_{L^{2}}^2
	&\le
	\sum_{\boldsymbol{x}_{h}\in \mathcal{G}}
	\abs{\Omega_{\boldsymbol{x}_{h}}}
	\sup\limits_{\bx\in  \Omega_{\boldsymbol{x}_{h}}}
	|\nabla \big(p_{h}u_{h}(\boldsymbol{x})\big)|
	\le
	6 \sum_{j=1}^3
	\|\widetilde{D}^{1,j} v_{h}\|_{\widetilde{L}_{h}^{2}}^2.
\end{split}
\end{equation*}
Hence $p_h$ satisfies the property \eqref{property 4} in Proposition \ref{define ph}.

Regarding property \eqref{property 5}, note that for any given $\boldsymbol{x}\in \Omega_{\boldsymbol{x}_{h}}$ where $\boldsymbol{x}_h\subset \mathcal{G}$, we can use the definitions of $p_{h}$ and $q_{h}$ to derive the following:
\begin{align*}
\abs{p_{h}u_{h}(\boldsymbol{x}) - q_{h}u_{h}(\boldsymbol{x})}
\le
h \sum_{j=1}^3 \abs{\widetilde{D}^{1,j} u_{h}}.
\end{align*}
Substituting it into the right-hand side of equation
\begin{align*}
\int_{\mathbb{R}^{2}}(p_{h}u_{h} - q_{h}u_{h})^{2}\mathrm{d}\boldsymbol{x}
=
\sum_{\boldsymbol{x}_{h}\in \mathcal{G}}
\int_{\Omega_{\boldsymbol{x}_{h}}}(p_{h}u_{h} - q_{h}u_{h})^{2}\mathrm{d}\boldsymbol{x},
\end{align*}
then one finally derives for $u_{h}\in \widetilde{H}_h^1$ that
\begin{align*}
\int_{\mathbb{R}^{2}}(p_{h}u_{h} - q_{h}u_{h})^{2}\mathrm{d}\boldsymbol{x}
\leq
h^{2}\sum_{j=1}^{3} \|\widetilde{D}^{1,j}u_{h}\|^{2}_{\widetilde{L}_{h}^{2}}.
\end{align*}
The proof of Propositon \ref{define ph} is completed.

\appendix

\section{Proof of Proposition \ref{thm: discrete pde equation}}\label{sec: A PDE type equation on lattice}

Firstly, it is easy to verify that $\bm_{h}(\boldsymbol{x}_{h},t)$ remains of unit length by multiplying equation (\ref{discrete LLG}) with $\bm_{h}$ and then deriving
\begin{equation}\label{m_h remain the unit length}
	\frac{\mathrm{d}}{\mathrm{d}t}|\bm_{h}(\boldsymbol{x}_{h},t)|^{2}=0.
\end{equation}

To prove the proposition, it suffices to prove that
\begin{equation}\label{new field}
	\frac{1}{(h^*)^{2}}\bm_{h}(\boldsymbol{x}_{h})\times\mathbf{B}_l(\boldsymbol{x}_{h})
	= \bm_{h}(\boldsymbol{x}_{h})\times\mathbf{B}^{*}_l(\boldsymbol{x}_{h}),
	\quad l = 1,2,3,4
\end{equation}
with $\mathbf{B}_l$ and $\mathbf{B}_l^*$ given in \eqref{expression of discrete effect field} and \eqref{new form of discrete effect field}, respectively. 
For the convenience of reading, we rewrite the equation \eqref{expression of discrete effect field} in the following:
\begin{equation}\label{expression of discrete effect field revisit}
	\left\{\begin{aligned}
			\mathbf{B}_1(\boldsymbol{x}_{h})
			= &
			(h^*)^2\sum_{k=1}^{3}\sum_{j=1}^{N_{k}}\frac{J_{k}}{h^2}\bm_{h}(\boldsymbol{x}_{h}^{k,j})\\
			\mathbf{B}_2(\boldsymbol{x}_{h})
			= &
			2(h^*)^2 \sum_{j=1}^{3}\frac{K}{h^2}
			\Big(\bm_{h}(\boldsymbol{x}_{h})\cdot\bm_{h}(\boldsymbol{x}_{h}^{1,j})\Big)\bm_{h}(\boldsymbol{x}_{h}^{1,j}),\\
			\mathbf{B}_3(\boldsymbol{x}_{h})
			= &
			\sum_{k=1}^{3}\sum_{j=1}^{N_{k}}\Big(L_{k}(\boldsymbol{x}_{h} ) + L_{k}(\boldsymbol{x}_{h}^{k,j})\Big)\Big(\bm_{h}(\boldsymbol{x}_{h}^{k,j})
			\cdot\mathbf{e}^{z}\Big)\mathbf{e}^{z}\\
			& +
			2\lambda(\boldsymbol{x}_{h})\Big(\bm_{h}(\boldsymbol{x}_{h})\cdot \mathbf{e}^{x}\Big) \mathbf{e}^{x},\\
			\mathbf{B}_4(\boldsymbol{x}_{h})
			= &
			\mu\mathbf{B}(\boldsymbol{x}_{h}).
	\end{aligned}\right.
\end{equation}
And the equation \eqref{new form of discrete effect field} is rewritten in \eqref{new form of discrete effect field revisit}.
Now let us calculate the left-hand side of \eqref{new field} by using \eqref{expression of discrete effect field revisit}. Note that
\begin{equation}\label{substituting x_i^kj}
	\bm_{h}(\boldsymbol{x}_{h}^{k,j}) = d_k h \widetilde{D}^{k,j} \bm_{h}(\boldsymbol{x}_{h}) + \bm_{h}(\boldsymbol{x}_{h}).
\end{equation}
Applying \eqref{substituting x_i^kj}, together with unit-length property of $\bm_{h}$, one has
\begin{align*}\label{B4}
	\frac{1}{(h^*)^{2}}\bm_{h}(\boldsymbol{x}_{h})\times\mathbf{B}_1(\boldsymbol{x}_{h})
	= &
	\sum_{k=1}^{3}\sum_{j=1}^{N_{k}}\frac{J_{k} d_k}{h}\bm_{h}(\boldsymbol{x}_{h}) \times \widetilde{D}^{k,j} \bm_{h}(\boldsymbol{x}_{h}),
	\end{align*}
and,
 \begin{align*}
	&\frac{1}{(h^*)^{2}}\bm_{h}(\boldsymbol{x}_{h})\times\mathbf{B}_2(\boldsymbol{x}_{h})\\
	= &
		2\sum_{j=1}^{3}\frac{K d_1}{h}
		\Big(1 + d_1 h\bm_{h}(\boldsymbol{x}_{h})\cdot \widetilde{D}^{1,j} \bm_{h}(\boldsymbol{x}_{h}) \Big)
		\bm_{h}(\boldsymbol{x}_{h})
		\times
		\widetilde{D}^{1,j} \bm_{h}(\boldsymbol{x}_{h}).
	\end{align*}
 Therefore combining them with \eqref{another form of delta} results in the equation \eqref{new field} when $l=1,2$. The second term in $\mathbf{B}_3$ and the term $\mathbf{B}_4$ is trivial. Now we consider the first term in $\mathbf{B}_3$ and denote it by
 \begin{equation*}
 	\mathbf{B}_3^*
 	=
 	\sum_{k=1}^{3}\sum_{j=1}^{N_{k}}\Big(L_{k}(\boldsymbol{x}_{h} ) + L_{k}(\boldsymbol{x}_{h}^{k,j})\Big)\Big(\bm_{h}(\boldsymbol{x}_{h}^{k,j})
 	\cdot\mathbf{e}^{z}\Big)\mathbf{e}^{z}.
 \end{equation*}
One can make the use of following transformation:
\begin{equation*}
    \begin{aligned}
        &\sum_{k=1}^{3}\sum_{j=1}^{N_{k}}L_{k}(\boldsymbol{x}_{h} ) \Big(\bm_{h}(\boldsymbol{x}_{h}^{k,j})
		\cdot\mathbf{e}^{z}\Big)\mathbf{e}^{z}\\
  = &
  \sum_{k=1}^{3} N_{k} L_{k}(\boldsymbol{x}_{h} ) \Big(\bm_{h}(\boldsymbol{x}_{h})
		\cdot\mathbf{e}^{z}\Big)\mathbf{e}^{z}
  + 
  \sum_{k=1}^{3}\sum_{j=1}^{N_{k}} d_k h L_{k}(\boldsymbol{x}_{h} ) \Big(\widetilde{D}^{k,j} \bm_{h}(\boldsymbol{x}_{h})
		\cdot\mathbf{e}^{z}\Big)\mathbf{e}^{z},
    \end{aligned}
\end{equation*}
and
\begin{equation*}
    \begin{aligned}
        &\sum_{k=1}^{3}\sum_{j=1}^{N_{k}}L_{k}(\boldsymbol{x}_{h}^{k,j} ) \Big(\bm_{h}(\boldsymbol{x}_{h}^{k,j})
		\cdot\mathbf{e}^{z}\Big)\mathbf{e}^{z}\\
  = &
  \sum_{k=1}^{3} N_{k} L_{k}(\boldsymbol{x}_{h} ) \Big(\bm_{h}(\boldsymbol{x}_{h})
		\cdot\mathbf{e}^{z}\Big)\mathbf{e}^{z}
  + 
  \sum_{k=1}^{3}\sum_{j=1}^{N_{k}} d_k h \Big(\widetilde{D}^{k,j} (L_{k}\bm_{h})(\boldsymbol{x}_{h})
		\cdot\mathbf{e}^{z}\Big)\mathbf{e}^{z}.
    \end{aligned}
\end{equation*}
Using the transformation above, together with \eqref{another form of delta}, it follows that
\begin{equation*}
    \begin{aligned}
        \mathbf{B}_3^*
  = &
  2 \sum_{k=1}^{3} N_{k} L_{k}(\boldsymbol{x}_{h} ) \Big(\bm_{h}(\boldsymbol{x}_{h})
		\cdot\mathbf{e}^{z}\Big)\mathbf{e}^{z}
  +
  \sum_{k=1}^{3} l_k h^2 L_{k}(\boldsymbol{x}_{h} ) \Big(\widetilde{\Delta}^{k} \bm_{h}(\boldsymbol{x}_{h})
		\cdot\mathbf{e}^{z}\Big)\mathbf{e}^{z}\\
  & + 
  \sum_{k=1}^{3} l_k h^2 \Big(\widetilde{\Delta}^{k} (L_{k}\bm_{h})(\boldsymbol{x}_{h})
		\cdot\mathbf{e}^{z}\Big)\mathbf{e}^{z}.
    \end{aligned}
\end{equation*}
Substitution of the expression above leads to equation \eqref{new field} for the case when $l=3$. 
Hence the proposition is proved.

\section{Discrete energy dissipation}
Before we prove Proposition \ref{thm: uniform estimate}, an energy dissipation property for the discrete equations \eqref{discrete LLG}-\eqref{new form of discrete effect field} is needed.
Instead of using the energy $\mathcal{H}_h[\bm_h]$ as given in \eqref{discrete energy}, we introduce another form of energy, denoted by $\mathcal{H}_h^*[\bm_{h}]$, which is obtained through a difference quotient
\begin{equation}\label{new form of energy}
	\begin{aligned}
		&\mathcal{H}_h^*[\bm_{h}]\\
		=&
		\sum_{k=1}^{3} J_k \frac{d_k^2}{4g_k^2} \Vert\widetilde{\nabla}_{k}\bm_{h}\Vert_{\widetilde{L}_{h}^{2}}^{2}
		+
		K \frac{d_1^2}{2g_1^2} \Vert\widetilde{\nabla}_{1}\bm_{h}\Vert_{\widetilde{L}_{h}^{2}}^{2}
		-
		K \frac{d_1^2}{2g_1^2} \Vert\bm_{h}\cdot\widetilde{\nabla}_{1}\bm_{h}\Vert_{\widetilde{L}_{h}^{2}}^{2}\\
		&+ \frac{h^2}{(h^*)^2} \sum_{k=1}^{3}
		\frac{d_k^2}{2g_k^2} \|  L_k^{1/2} \, 
  \widetilde{\nabla}^{k} \bm_{h}\cdot\mathbf{e}^{z}\|_{\widetilde{L}_{h}^{2}}^{2}\\
    &+ \frac{h^2}{(h^*)^2}
\sum_{k=1}^{3}\sum_{j=1}^{N_k} 
		\frac{d_k^2}{2}\Big(  
  (\widetilde{D}^{k,j} L_k) \big( \bm_{h}\cdot\mathbf{e}^{z}\big)
  \mathbf{e}^{z}
  ,\,
  \widetilde{D}^{k,j}
  \bm_{h} \Big)_h\\
		& -
		\frac{1}{(h^*)^2} \sum_{k=1}^{3} N_k \Vert L_k^{1/2}\,  \bm_{h}\cdot \mathbf{e}^{z}\Vert_{\widetilde{L}_{h}^{2}}^{2}
		- 
  \frac{1}{(h^*)^2} \Vert \lambda^{1/2}\, \bm_{h}\cdot \mathbf{e}^{x}\Vert_{\widetilde{L}_{h}^{2}}^{2}
		-
		\frac{\mu}{(h^*)^2} \big(\bm_{h},\, \mathbf{B}\big)_h.
	\end{aligned}
\end{equation}
Then we have:
\begin{proposition}\label{thm: energy inequality}
	Let $\bm_{h}$ be the solution to the equations \eqref{discrete LLG}-\eqref{new form of discrete effect field}. The following energy dissipation holds:
	\begin{equation}\label{energy ineq}
		\begin{split}
			\mathcal{H}_h^*[\bm_{h}(t)]
			+
			\alpha \int_0^t \| \bm_{h}(s)\times\mathbf{B}^*_{\rm{eff},h}(s)\|_{\widetilde{L}_{h}^{2}}^{2} \d s \le \mathcal{H}_h^*[\bm_{h}(0)],
		\end{split}
	\end{equation}
	holds for every $t\in [0, +\infty)$, where the energy $\mathcal{H}_h^*$ is given by \eqref{new form of energy}.
\end{proposition}
\begin{proof}
	Since
	\begin{align*}
		&\Big(\bm_{h}\times(\bm_{h}\times\mathbf{B}^*_{\rm{eff},h})\Big) \cdot \mathbf{B}^*_{\rm{eff},h}
		=-\big| \bm_{h}\times\mathbf{B}^*_{\rm{eff},h}\big|^{2},
	\end{align*}
	taking the $\widetilde{L}_{h}^{2}$ inner-product of  \eqref{discrete LLG} and $\mathbf{B}^*_{\rm{eff},h}$, one has
	\begin{align}\label{m_t cdot effective field}
		- \Big(\frac{\d\bm_{h}}{\d t},\,\, \mathbf{B}^*_{\rm{eff},h}\Big)_h
		+
		\alpha \| \bm_{h}(t)\times\mathbf{B}^*_{\rm{eff},h}(t)\|_{\widetilde{L}_{h}^{2}}^{2} = 0.
	\end{align}
We assert that
\begin{equation}\label{partial_t m times B}
	- \Big(\frac{\d\bm_{h}}{\d t},\,\, \mathbf{B}^*_{\rm{eff},h}\Big)_h
	=
	\frac{\mathrm{d}}{\mathrm{d}t} \mathcal{H}_h^*[\bm_{h}].
\end{equation}
Then the inequality \eqref{energy ineq} follows by substituting \eqref{partial_t m times B} into \eqref{m_t cdot effective field} and integration over $[0,t]$.

Now let us compute left-hand side of \eqref{partial_t m times B} term by term by applying \eqref{new form of discrete effect field}. For the convenience of reading, we rewrite the equation \eqref{new form of discrete effect field} here:
\begin{equation}\label{new form of discrete effect field revisit}
	\left\{\begin{aligned}
		&\mathbf{B}_1^*(\boldsymbol{x}_{h})
		= 
		\sum_{k=1}^{3} l_{k}J_k\widetilde{\Delta}_{k}\bm_{h}(\boldsymbol{x}_{h})\\
		&\mathbf{B}_2^*(\boldsymbol{x}_{h})
		= 
		2 l_1 K \widetilde{\Delta}_{1}\bm_{h}(\boldsymbol{x}_{h})
		+
		2K\sum_{j=1}^{3}\Big(\bm_{h}(\boldsymbol{x}_{h})\cdot \widetilde{D}^{1,j}\bm_{h}(\boldsymbol{x}_{h})\Big)\widetilde{D}^{1,j}\bm_{h}(\boldsymbol{x}_{h}),\\
		&\mathbf{B}^{(1)}_3(\boldsymbol{x}_{h})
		= 
		\frac{2}{(h^*)^2} \lambda (\boldsymbol{x}_{h}) \Big(\bm_{h}(\boldsymbol{x}_{h})\cdot \mathbf{e}^{x}\Big) \mathbf{e}^{x}
		+
		\frac{2}{(h^*)^2} \sum_{k=1}^{3} N_k L_k (\boldsymbol{x}_{h}) \Big(\bm_{h}(\boldsymbol{x}_{h})\cdot\mathbf{e}^{z}\Big)
		\mathbf{e}^{z},\\
		&\mathbf{B}^{(2)}_3(\boldsymbol{x}_{h})
		=
		\frac{h^2}{(h^*)^2} \sum_{k=1}^{3}
		l_{k} \Big\{
		L_k (\boldsymbol{x}_{h}) \Big(\widetilde{\Delta}_{k}\bm_{h}(\boldsymbol{x}_{h})\cdot\mathbf{e}^{z}\Big)\mathbf{e}^{z}
		+
		\Big(\widetilde{\Delta}_{k} (L_k\bm_{h})(\boldsymbol{x}_{h})\cdot\mathbf{e}^{z}\Big)\mathbf{e}^{z}\Big\}\\
		&\mathbf{B}^*_4(\boldsymbol{x}_{h})
		= 
		\frac{\mu}{(h^*)^2}\mathbf{B}(\boldsymbol{x}_{h}).
	\end{aligned}\right.
\end{equation}
For the term $\mathbf{B}^*_{1}$ in \eqref{new form of discrete effect field revisit}, we apply \eqref{Green fomular} and deduce that
\begin{equation*}
	\begin{aligned}
	-\Big(\frac{\d\bm_{h}}{\d t},\ \mathbf{B}^*_{1} \Big)_h
	=&
	-\sum_{k=1}^{3} J_k\frac{d_k^2}{2}\Big(\frac{\d\bm_{h}}{\d t},\ \widetilde{D}^{k,j} \big\{ \widetilde{D}_{B}^{k,j} \bm_{h} \big\}  \Big)_h\\
 =&
	\sum_{k=1}^{3} J_k\frac{d_k^2}{2}\Big( \widetilde{D}^{k,j} 
 \frac{\d\bm_{h}}{\d t},\ \widetilde{D}^{k,j} \bm_{h} \Big)_h
	=
 \frac{1}{2}\frac{\mathrm{d}}{\mathrm{d}t}\sum_{k=1}^{3} J_k \frac{d_k^2}{2g_k^2} \Vert\widetilde{\nabla}_{k}\bm_{h}\Vert_{\widetilde{L}_{h}^{2}}^{2},
	\end{aligned}
	\end{equation*}
where we used \eqref{define nabla} and Lemma \ref{integration by parts for one step}. The first term in $\mathbf{B}^*_{2}$ can be checked in a similar way.
For the second term in $\mathbf{B}^*_{2}$, we assert that
\begin{equation}\label{difficult term in energy ineq}
	\begin{aligned}
	&- \sum_{j=1}^3\Big(\frac{\d\bm_{h}}{\d t},\,\, \big(\bm_{h}\cdot \widetilde{D}^{1,j}\bm_{h}\big)\widetilde{D}^{1,j}\bm_{h}\Big)_h\\
	= &
	-\frac{d_1^2}{4g_1^2} \frac{\d}{\d t} \big\Vert\bm_{h}(\boldsymbol{x}_{h})\cdot \widetilde{\nabla}_{1}\bm_{h}(\boldsymbol{x}_{h})\big\Vert_{\widetilde{L}_{h}^{2}}^2.
	\end{aligned}
\end{equation}
To clarify, we can calculate the inner product on the left-hand side of equation \eqref{difficult term in energy ineq} for the interior and gather all terms that involve points $\boldsymbol{x}_{h}$ and its neighbor $\boldsymbol{x}_{h}^{1,j}$. This yields two resulting terms:
\begin{equation}\label{term 1 in energy ineq}
	-\Big(\bm_{h}(\boldsymbol{x}_{h})\cdot \widetilde{D}^{1,j}\bm_{h}(\boldsymbol{x}_{h})\Big)
	\cdot
	\Big(\frac{\d\bm_{h}}{\d t}(\boldsymbol{x}_{h})
	\cdot
	\widetilde{D}^{1,j}\bm_{h}(\boldsymbol{x}_{h})\Big),
\end{equation}
and
\begin{equation}\label{term 2 in energy ineq}
	-\Big(\bm_{h}(\boldsymbol{x}_{h}^{1,j})\cdot \widetilde{D}^{1,j}\bm_{h}(\boldsymbol{x}_{h}^{1,j})\Big)
	\cdot
	\Big(\frac{\d\bm_{h}}{\d t}(\boldsymbol{x}_{h}^{1,j})
	\cdot
	\widetilde{D}^{1,j}\bm_{h}(\boldsymbol{x}_{h}^{1,j})\Big).
\end{equation}
By using exchange of notation in \eqref{exchange express of difference} and unit-length property, adding \eqref{term 1 in energy ineq} and \eqref{term 2 in energy ineq} together equals
\begin{equation*}
	-\frac{1}{2} \frac{\d}{\d t} \big\vert\bm_{h}(\boldsymbol{x}_{h})\cdot \widetilde{D}^{1,j}\bm_{h}(\boldsymbol{x}_{h})\big\vert^2.
\end{equation*}
Then equation \eqref{difficult term in energy ineq} follows by the notation in \eqref{define nabla}. 

Finally, one can check that the computation of terms $\mathbf{B}^{(1)}_{3}$ and $\mathbf{B}^*_{4}$ is trivial. As for the term $\mathbf{B}^{(2)}_{3}$, we can deduce from \eqref{define nabla} and Lemma \ref{integration by parts for one step}:
\begin{equation*}
    \begin{aligned}
		l_{k} \Big( 
  L_k(\boldsymbol{x}_{h}) &
  \big(\widetilde{\Delta}_{k}\bm_{h}(\boldsymbol{x}_{h})\cdot\mathbf{e}^{z}\big)
  \mathbf{e}^{z}
  ,\,
  \frac{\d}{\d t}\bm_{h}(\boldsymbol{x}_{h}) \Big)_h \\
  = & -
	\sum_{j=1}^{N_k} 
		\frac{d_k^2}{2}\Big(  
  \big(\widetilde{D}^{k,j} \bm_{h}(\boldsymbol{x}_{h})\cdot\mathbf{e}^{z}\big)
  \mathbf{e}^{z}
  ,\,
  \widetilde{D}^{k,j} \big\{L_k(\boldsymbol{x}_{h}) 
  \frac{\d}{\d t} \bm_{h}(\boldsymbol{x}_{h})\big\} \Big)_h,
    \end{aligned}
\end{equation*}
and,
\begin{equation*}
    \begin{aligned}
		l_{k} \Big( 
  \big( \widetilde{\Delta}_{k} (L_k&\bm_{h})(\boldsymbol{x}_{h}) \cdot\mathbf{e}^{z}\big) \mathbf{e}^{z}
  ,\,
  \frac{\d}{\d t}\bm_{h}(\boldsymbol{x}_{h}) \Big)_h\\
  = & -
	\sum_{j=1}^{N_k} 
		\frac{d_k^2}{2} \Big(  
  \big(\widetilde{D}^{k,j} (L_k\bm_{h})(\boldsymbol{x}_{h})\cdot\mathbf{e}^{z}\big)
  \mathbf{e}^{z}
  ,\,
  \widetilde{D}^{k,j}
  \frac{\d}{\d t} \bm_{h}(\boldsymbol{x}_{h}) \Big)_h.
    \end{aligned}
\end{equation*}
Therefore summing up one can derive
\begin{equation*}
    \begin{aligned}
-\Big(\frac{\d\bm_{h}}{\d t},\ \mathbf{B}^{(2)}_{3} \Big)_h
  = &  \frac{h^2}{(h^*)^2} \frac{d_k^2}{2}
	\frac{\d}{\d t}  \sum_{j=1}^{N_k} 
		 \Big(  
  \big(\widetilde{D}^{k,j} (L_k\bm_{h})(\boldsymbol{x}_{h})\cdot\mathbf{e}^{z}\big)
  \mathbf{e}^{z}
  ,\,
  \widetilde{D}^{k,j}
  \bm_{h}(\boldsymbol{x}_{h}) \Big)_h,
    \end{aligned}
\end{equation*}
where the right-hand side can be written as
\begin{equation*}
 \frac{h^2}{(h^*)^2} \frac{d_k^2}{2} \bigg\{
	\frac{\d}{\d t} 
		\frac{1}{g_k^2} \|  L_k^{1/2} \,
  \widetilde{\nabla}^{k} \bm_{h}(\boldsymbol{x}_{h})\cdot\mathbf{e}^{z}\|_{\widetilde{L}_{h}^{2}}^{2}
    +  
	\frac{\d}{\d t} \sum_{j=1}^{N_k} 
		\Big(  
  (\widetilde{D}^{k,j} L_k) \big( \bm_{h}\cdot\mathbf{e}^{z}\big)
  \mathbf{e}^{z}
  ,\,
  \widetilde{D}^{k,j}
  \bm_{h} \Big)_h \bigg\}.
\end{equation*}
This completes the proof.
\end{proof}

\section{Proof of Proposition \ref{thm: uniform estimate}}\label{sec:estimation}
The equation \eqref{m_h remain the unit length} implies that $|\bm_{h}|=1$, and the energy inequality \eqref{energy ineq} deduce
	\begin{align}\label{energy inequality}
		\mathcal{H}_h^*[\bm_{h}(t)] \le \mathcal{H}_h^*[\bm_{h}(0)].
	\end{align}
 Notice that using the assumption \ref{assume} and $|\bm_{h}|=1$, we have the estimate:
 \begin{equation}
     \begin{aligned}
		\frac{1}{(h^*)^2} \sum_{k=1}^{3} N_k \Vert L_k^{1/2} \bm_{h}\cdot \mathbf{e}^{z}\Vert_{\widetilde{L}_{h}^{2}}^{2}
		+
  \frac{1}{(h^*)^2} \Vert \lambda^{1/2} \bm_{h}\cdot \mathbf{e}^{x}\Vert_{\widetilde{L}_{h}^{2}}^{2}
		+
		\frac{\mu}{(h^*)^2} \big(\bm_{h},\, \mathbf{B}\big)_h
  \le C,
     \end{aligned}
 \end{equation}
 where constant $C$ depends only on  $\Vert \mathbf{B} \Vert_{L^1}$, $\|L\|_{L^1}$ and $\|\lambda\|_{L^1}$.
	Thus applying \eqref{new form of energy}-\eqref{energy ineq}, it leads to
 \begin{equation*}
     \begin{aligned}
         \sum_{k=1}^{3} J_k \frac{d_k^2}{4g_k^2} \Vert\widetilde{\nabla}_{k}\bm_{h}\Vert_{\widetilde{L}_{h}^{2}}^{2}
		+
		K \frac{d_1^2}{2g_1^2} \Big(\Vert\widetilde{\nabla}_{1}\bm_{h}\Vert_{\widetilde{L}_{h}^{2}}^{2}
		-\Vert\bm_{h}\cdot\widetilde{\nabla}_{1}\bm_{h}\Vert_{\widetilde{L}_{h}^{2}}^{2}\Big)& \\
		+ \frac{h^2}{(h^*)^2} \sum_{k=1}^{3}
		\frac{d_k^2}{2g_k^2} \|  L_k^{1/2} \, 
  \widetilde{\nabla}^{k} \bm_{h}\cdot\mathbf{e}^{z}\|_{\widetilde{L}_{h}^{2}}^{2}
  &\le
			\mathcal{H}_h^*[\bm_{h}(0)]
			+ C.
     \end{aligned}
 \end{equation*}
	Here, the second term on left-hand side is non-negative. Hence,
	\begin{equation*}
		 \sum_{k=1}^{3} J_k \frac{d_k^2}{4g_k^2} \Vert\widetilde{\nabla}_{k}\bm_{h}\Vert_{\widetilde{L}_{h}^{2}}^{2}
		\le
		\mathcal{H}_h^*[\bm_{h}(0)] + C,
	\end{equation*}
	which leads to the second inequality in \eqref{bound of gradient}.

On the other hand, the equation \eqref{discrete LLG} implies that
	\begin{equation*}
		\int_{0}^{T} \left\|\frac{\partial\bm_{h}}{\partial t} \right\|_{\widetilde{L}_{h}^{2}}^2 \d t
		= (1 + \alpha^2)
		\int_{0}^{T} \|\bm_{h}\times\mathbf{B}^*_{\rm{eff},h}\|_{\widetilde{L}_{h}^{2}}^2 \d t.
	\end{equation*}
Hence, the inequality \eqref{thm: uniform estimate} is also proved, by virtue of Proposition
 \ref{thm: energy inequality}. The proof of Proposition \ref{thm: uniform estimate} is completed.

\section*{Acknowledgments}
This work was supported by NSFC grant 11971021 (J. Chen). The work of Zhiwei Sun is supported by China Scholarship Council grant 202006920083.

\bibliographystyle{amsplain}
\bibliography{reference}

\end{document}